\def\arxivversion{}
\ifdefined\arxivversion
\documentclass[acmsmall,nonacm]{acmart}
\else
\documentclass[acmsmall,anonymous,screen,review]{acmart}
\fi
\AtBeginDocument{%
  }

\setcopyright{acmlicensed}
\copyrightyear{2026}
\acmYear{2026}
\acmDOI{XXXXXXX.XXXXXXX}
\acmConference[FSE '27]{ACM International Conference on the Foundations of Software Engineering}{July 12--16,
  2027}{Shenzhen, China}
\acmISBN{978-1-4503-XXXX-X/2027/07}

\usepackage[T1]{fontenc}
\usepackage{amsmath,amsfonts}
\allowdisplaybreaks
\usepackage{mathtools}

\usepackage{graphicx}
\usepackage{textcomp}
\usepackage{xcolor}
\usepackage{stmaryrd}
\usepackage{listings}
\usepackage{xspace}
\usepackage{algorithm}
\usepackage{wasysym}
\usepackage{algpseudocode}
\usepackage{mathpartir}
\usepackage{tikz}
\usepackage{booktabs}
\usepackage{paralist}
\usepackage{comment}
\usepackage{hyperref}
\usepackage{cleveref}
\usepackage{dot2texi}

\newcommand{\rot}[1]{\rotatebox{90}{\shortstack[l]{#1}}}
\newif\ifextended
\ifdefined\conferenceversion\extendedfalse\else\extendedtrue\fi
\ifdefined\arxivversion
\newcommand{\repourl}{https://github.com/runtime-enforcement/EnfFlash}
\else
\newcommand{\repourl}{https://anonymous.4open.science/r/EnfFlash-6400}
\fi

\usetikzlibrary{arrows.meta,positioning,fit,backgrounds}

\newcommand{\leanicon}{\raisebox{-0.15ex}{\includegraphics[height=1.5ex]{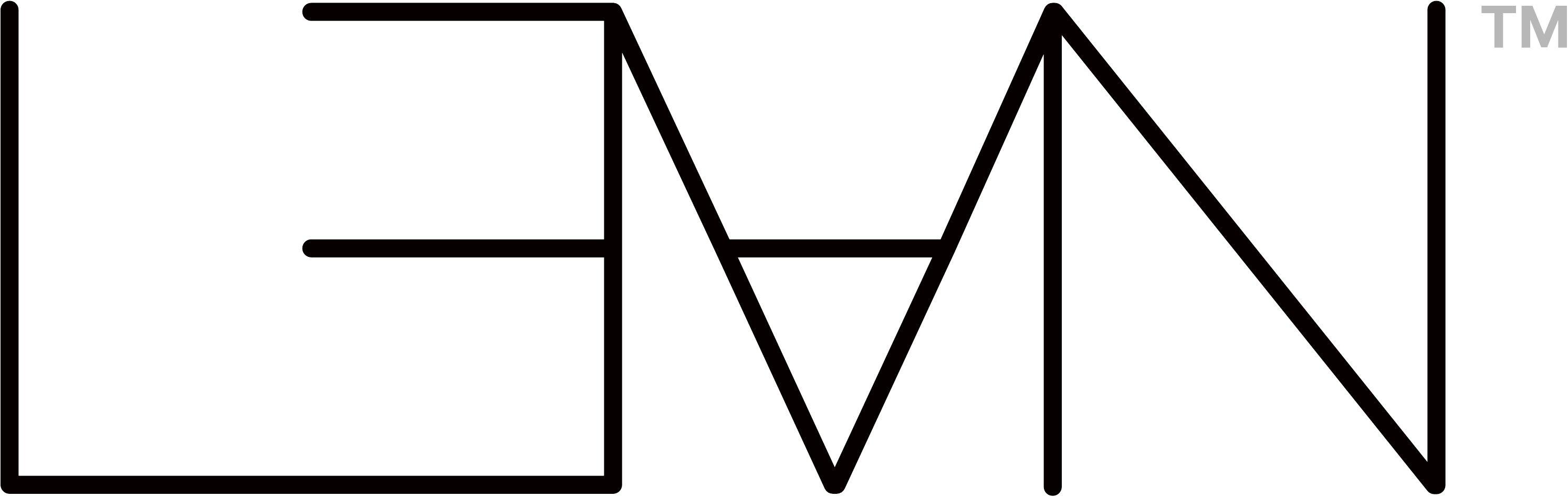}}}
\newcommand{\leanref}[1]{{\upshape\unskip\nobreak\hfill\penalty50\hskip1em\null\nobreak\hfill
  \mbox{[\,\leanicon\;\texttt{\small\detokenize{#1}}\,]}}}

\newcommand{\EnfFlash}{\textsc{EnfFlash}}
\newcommand{\EnfGuard}{\textsc{EnfGuard}}
\newcommand{\WhyEnf}{\textsc{WhyEnf}}
\newcommand{\MonPoly}{\textsc{MonPoly}}
\newcommand{\EnfPoly}{\textsc{EnfPoly}}

\newcommand{\underv}[1]{\llbracket\, #1 \,\rrbracket_{v}}
\newcommand{\Sequence}[2][]{\left\langle{#2}\right\rangle_{#1}}
\newcommand{\Pow}[1]{\mathcal{P}(#1)}

\newcommand{\Databases}{\mathbb{DB}}

\newcommand{\Variables}{\mathbb{V}}
\newcommand{\Traces}{\mathbb{T}}
\newcommand{\Intervals}{\mathbb{I}}

\newcommand{\emptytrace}{\varepsilon}

\newcommand{\sem}[1]{\llbracket #1 \rrbracket}

\newcommand{\use}{\mathsf{use}}
\newcommand{\consent}{\mathsf{consent}}
\newcommand{\delete}{\mathsf{delete}}
\newcommand{\deleteReq}{\mathsf{deletion\_request}}

\newcommand{\revoke}{\mathsf{revoke}}

\begin{document}

\title{Practical Runtime Enforcement of First-Order Temporal Requirements}

\author{François Hublet}
\email{francois.hublet@inf.ethz.ch}
\affiliation{%
  \department{Department of Computer Science}
  \institution{ETH Zürich}
  \city{Zurich}
  \country{Switzerland}}

\author{David Basin}
\email{basin@inf.ethz.ch}
\affiliation{%
  \department{Department of Computer Science}
  \institution{ETH Zürich}
  \city{Zurich}
  \country{Switzerland}}

\author{Srđan Krstić}
\email{srdan.krstic@di.ku.dk}
\affiliation{%
  \department{Department of Computer Science}
  \institution{University of Copenhagen}
  \city{Copenhagen}
  \country{Denmark}}

\renewcommand{\shortauthors}{Hublet, Basin, and Krstić}

\newcommand{\seq}[1]{\langle #1 \rangle}
\definecolor{kwcolor}{rgb}{0.09,0.37,0.65}
\newcommand{\kw}[1]{{\color{kwcolor}\text{\ttfamily\bfseries#1}}}
\newcommand{\kwc}[1]{\text{\ttfamily\bfseries#1}}

\newcommand{\Events}{\mathbb{E}}
\newcommand{\Domain}{\mathbb{D}}
\newcommand{\Functions}{\mathbb{F}}
\newcommand{\mset}[1]{\langle\!\langle #1 \rangle\!\rangle}
\newcommand{\Nat}{\mathbb{N}}
\newcommand{\guardable}{\mathop{\mathrm{Guardable}}}

\newcommand{\tp}{\mathsf{tp}}
\newcommand{\Cau}{\mathbb{C}}
\newcommand{\Sup}{\mathbb{S}}
\newcommand{\Obs}{\mathbb{O}}
\newcommand{\fv}{\mathrm{fv}}
\newcommand{\clause}[3]{\langle #1\,|\,#2\,|\,#3\rangle}

\newcommand*{\Prev}{\mathop{\CIRCLE}\nolimits}
\newcommand*{\Next}{\mathop{\Circle}\nolimits}
\newcommand*{\Since}{\mathbin{\mathsf{S}}}
\newcommand*{\SinceP}{\mathbin{\hat{\mathsf{S}}}}
\newcommand*{\UUntil}{\mathbin{\mathsf{U}}}
\newcommand*{\Always}{\operatorname{\square}}
\newcommand*{\PAlways}{\operatorname{\blacksquare}}
\newcommand*{\Eventually}{\operatorname{\lozenge}}
\newcommand*{\Once}{\operatorname{\blacklozenge}}
\newcommand*{\iin}{\mathsf{in}}
\newcommand*{\LLet}[3]{\operatorname{\llet}\,#1=#2\,\mathbin{\iin}#3}
\newcommand{\llet}{\textsf{let}}
\newcommand*{\Agg}[4]{\operatorname{\mathtt{#1}}(#2; #3)~#4}
\newcommand*{\LetAgg}[5]{#1 \gets \Agg{#2}{#3}{#4}{#5}}

\algnewcommand\algorithmicswitch{\textbf{switch}}
\algnewcommand\algorithmiccase{\textbf{case}}
\algnewcommand\algorithmicdefault{\textbf{default}}
\algdef{SE}[SWITCH]{Switch}{EndSwitch}[1]{\algorithmicswitch\ #1\ \algorithmicdo}{\algorithmicend\ \algorithmicswitch}%
\algdef{SE}[CASE]{Case}{EndCase}[1]{\algorithmiccase\ #1}{\algorithmicend\ \algorithmiccase}%
\algdef{SE}[DEFAULT]{Default}{EndDefault}{\algorithmicdefault}{\algorithmicend\ \algorithmiccase}%
\algtext*{EndSwitch}%
\algtext*{EndCase}%
\algtext*{EndDefault}%

\begin{abstract}

\looseness=-1
Runtime enforcers observe a system's behavior
and exert control over it to ensure that the system 
always adheres to its requirements.
Many natural requirements not only formulate restrictions on the system's present behavior, 
but also impose obligations on its future behavior; for instance, agentic security may require personal data collected during a run to be erased by some deadline.
In an ever-growing compliance landscape, 
real-world systems may be subject to hundreds of such requirements.
However, existing enforcement mechanisms supporting complex policies are often too slow for 
production software; instead, developers may resort to non-temporal access control mechanisms 
and ad-hoc instrumentation, but these scale poorly with requirement complexity.

To address this problem, we introduce an efficient algorithm and tool
for enforcing complex temporal requirements and demonstrate its
performance.
Specifically, we identify a fragment of Metric First-Order Temporal
Logic (MFOTL) that can be enforced with low runtime
complexity while supporting a rich family of practically relevant requirements, including obligations.
We then design an enforcement algorithm for this fragment and
implement it in \EnfFlash{}, a novel tool that enforces
requirements with latency up to $44\times$ lower than \EnfGuard{}, the previous state-of-the-art
enforcer, on standard benchmarks, and up to three orders of magnitude lower on security
policies for LLM agents. This performance is achieved by compiling
requirements into imperative programs that we efficiently interpret.
We evaluate \EnfFlash{} both standalone, on existing benchmarks, and integrated in web applications
as an enforcement backend. We demonstrate that it can enforce a substantial 400-line  
MFOTL formula specifying the GDPR on a social network while adding under 15\,ms to each page view, which is
sufficient for most interactive and real-time applications.
\end{abstract}



\maketitle

\section{Introduction}
\label{sec:intro}

Modern systems must comply with increasingly complex requirements and regulations.
Rather than auditing these requirements ex post, in \emph{runtime enforcement}~\cite{Schneider2000,Bauer2002} 
an \emph{enforcer} observes the system's behavior and ensures compliance at runtime by 
suppressing or causing system actions.

Requirements can be broadly classified as \emph{provisions}, whose satisfaction depends on past 
and present system behavior, and \emph{obligations}, whose satisfaction depends 
on its future behavior~\cite{HiltyBP05}. An example provision is enforcing access to confidential data only when 
a prior approval was granted, whereas an obligation would be that any temporary copy made  
of confidential data must be deleted within a given deadline. Both of these requirements are \emph{temporal} and \emph{first-order} (i.e., they relate events that carry data, such as file identifiers, across time). However, they require 
different enforcement capabilities: an enforcer can simply \emph{react} to
an attempted system action to enforce a provision, while it must \emph{proactively} act to enforce an obligation.
Runtime enforcement plays a crucial role in numerous domains such as AI agent security:
agent behavior cannot be statically analyzed, but 
only observed, e.g., as a sequence of tool invocations. The example provision and obligation above readily apply in 
this context with an enforcer mediating tool invocations to ensure policy compliance
and retaining responsibility for pending obligations independently of the agent's subsequent behavior. 

\looseness=-1
Recent industrial developments reinforce the relevance of temporal first-order requirements for 
agentic software and put forward enforcement performance as a key challenge for its practical adoption. 
AWS's Cedar access-control (AC) mechanism~\cite{Cutler2024} is a fast, non-temporal, first-order 
enforcer designed for a very limited set of provisions:
those that depend only on the current system state without considering its past 
behavior. To alleviate this limitation, AWS recently proposed Dogwood~\cite{dogwood}, which combines Cedar with 
Metric First-Order Temporal Logic (MFOTL)~\cite{Chomicki1995} to express provisions.
While MFOTL offers both past and future temporal operators to express both complex provisions and obligations,
Dogwood only supports provisions, with obligations identified as future work. 

Other enforcers based on MFOTL, like \EnfGuard{}~\cite{Hublet2025} and \WhyEnf{}~\cite{Hublet2024} support both
provisions and obligations. Unfortunately, these tools' latency is high and can grow with the size of the
enforcer's state. This leads to latencies of the order of seconds even for relatively simple requirements, 
making these tools unsuitable for interactive and real-time applications. 

\looseness=-1
An alternative to off-the-shelf enforcers is to directly implement ad-hoc, specialized ones to achieve 
good performance. However, this is brittle and scales poorly. Abandoning declarative specifications of provisions and 
obligations, underpinned with formal semantics, introduces a new problem: justifying that the ad-hoc implementation 
behaves as intended requires reasoning about the much more complex and often unavailable formal semantics 
of the implementation language. This, combined with the fact that in practice one must implement hundreds of  
interacting requirements and encode their temporal dependencies by hand, makes this error-prone and hard to evolve~\cite{Nguyen2026}.

\looseness=-1
To improve upon the status quo, software developers need a tool
that combines three properties: 
declarative specifications of data-dependent provisions and obligations, 
sound generation of reactive and proactive enforcement, and 
runtime performance efficient enough for interactive applications,
which require response times on the order of fractions of a second.

In this paper, we introduce \EnfFlash{}, a novel tool
for enforcing first-order temporal provisions and obligations with low, predictable latency.
Unlike previous tools, which \emph{interpret} MFOTL formulae directly~\cite{Hublet2022,Hublet2024,Hublet2025},
\EnfFlash{} \emph{compiles} each formula into a program of a small imperative language, EF,
that describes an efficient enforcer for that formula,
and then executes it.
The generated programs are provably correct, and their execution complexity can be statically bounded.
The MFOTL fragment supported by \EnfFlash{} is large enough to express a large collection of
provisions and obligations that arise in practice, including
requirements involving aggregations and user-defined operators.
Additionally, \EnfFlash{} lifts the key limitation of previous enforcers~\cite{Hublet2022,Hublet2024,Hublet2025} that
could not simultaneously cause and suppress the same event with
different parameters; its compiler optimizes the order of rule evaluation
and fixpoint computations to reduce overhead while preserving soundness.

We evaluate \EnfFlash{}'s expressivity and performance using existing benchmarks~\cite{Cutler2024,Hublet2025}
and deploy it as an enforcement backend for two mid-sized web applications: a social
network platform instrumented for compliance with a 400-line specification
of the General Data Protection Regulation (GDPR), and an event management platform
accompanied by a broad range of access-control rules. We also enforce security
policies on the tool calls of an LLM banking agent under prompt-injection attacks.
Compared with \EnfGuard{}~\cite{Hublet2025}, the previous state-of-the-art MFOTL enforcer,
\EnfFlash{}'s latency is 3--44$\times$ lower on 17 of the 20 benchmark formulae both process
(\EnfGuard{} times out on 19 others) and 20--5,570$\times$ lower on the agent's policies;
in web applications, enforcement adds 1--15\,ms per page view, compatible with production use.

\medskip\noindent
After reviewing enforcement preliminaries (Section \ref{sec:preliminaries}),
we make the following contributions:
\begin{itemize}
\item We design an imperative DSL called EF that captures
  the behavior of proactive first-order enforcers, give it a formal semantics,
  and bound EF's evaluation time complexity (Section~\ref{sec:ef}).
\item We introduce a novel compilation procedure from MFOTL to EF, 
  generating sound enforcers for an expressive fragment of MFOTL (Section~\ref{sec:mfotl-to-ef}).
\item We implement our algorithms in \EnfFlash{},
  and evaluate it both on existing enforcement
  benchmarks and as an enforcement backend for web applications and LLM agents.
  We show that \EnfFlash{} enforces specifications comprising hundreds of requirements with
  modest latency, with $3$--$44\times$ lower latency than \EnfGuard{} on most benchmark formulae
  (Section~\ref{sec:eval}).
\end{itemize}
Overall, \EnfFlash{} is the first enforcer for first-order temporal provisions and obligations
offering practical performance (Section~\ref{sec:rw}).
All theorems marked with~\leanicon{} are formally verified in Lean 4~\cite{Moura2021} for the
core logic presented in this paper; the complexity results in Section~\ref{sec:complexity} are proved on paper.
\ifextended\else
The extended version of this paper~\cite{extended} additionally characterizes the
enforceable fragment by a type system and provides more detailed evaluation results.
\fi

\section{Runtime Enforcement}

\label{sec:preliminaries}






We introduce traces that model executions,
metric-first order temporal logic (MFOTL), and enforcers.


\subsection{Traces} 
\looseness=-1
Let $\Events$ be a finite set of \emph{event names} with arity function $\iota: \Events \rightarrow
\Nat$ and $\Domain$ be a non-empty set of \emph{data values}.
A \emph{signature} $\Sigma = (\Domain, \Events, \iota)$ defines the set of \emph{databases}
$\Databases \coloneqq \Pow{\{(e,\bar{d}) \mid e \in \Events, \bar{d} \in \Domain^{\iota(e)}\}}.$
For $E \subseteq \Events$, let $\Databases_E \coloneqq \{D \in \Databases \mid \forall (e,\bar{d}) \in D.\ e \in E\}$ denote databases with event names in $E$.
\begin{example}\label{ex:gdpr}
  Consider a file management system. Files can be
  accessed (`used') or deleted, and users can request deletion of their data,
  give consent for files of some category to be used, or revoke said consent.
  We model this with the signature $\Sigma_0=(\Nat,\Events_0,\iota_0)$, where
  $\Events_0=\{\use, \consent,\allowbreak \revoke,\allowbreak \delete,\allowbreak \deleteReq\}$,
  $\iota_0(\use)=3$, and $\iota_0(\delete)=\iota_0(\deleteReq)=\iota_0(\consent)=\iota_0(\revoke)=2$.
  The event
  $\use(c,d,u)$ denotes `system uses user $u$'s data $d$ from category $c$',
  $\delete(d,u)$ denotes `user $u$'s data $d$ is deleted',
  $\deleteReq(d,u)$ denotes `user $u$ requests deletion of data $d$',
  $\consent(u,c)$ denotes `user $u$ provides consent for category $c$', and
  $\revoke(u,c)$ denotes `user $u$ revokes consent for category $c$'.
\end{example}
A \emph{trace} $\sigma$ is a finite or infinite
sequence $\Sequence[0\leq i < k]{(\tau_i,D_i)}$ of length $|\sigma| = k
\in \Nat \cup \lbrace \infty\rbrace$ of pairs of timestamps $\tau_i \in \Nat$ and databases $D_i \in\Databases$, where timestamps grow \emph{monotonically}
($\tau_i \leq \tau_{i+1}$ for all $i$ with $i+1 < |\sigma|$) and \emph{progress} (if $|\sigma| = \infty$, then $\forall \tau.\; \exists i.\;\tau < \tau_i$).
An index $0\leq i<|\sigma|$ in a trace $\sigma$ is called a \emph{time-point}.
The empty trace is denoted by $\emptytrace$, the set of all traces by $\Traces$.

\begin{example}\label{ex:traces}
  Consider two traces of the above system:
  \begin{align*}
    \sigma_1&=(10, \{ \consent(1, 1) \}), (50, \{ \use(1, 3, 1), \use(2, 1, 1) \}), ... \\
    \sigma_2&=(10, \{ \deleteReq(1, 1)\}), (50, \{ \use(1, 3, 1) \}), ...
  \end{align*}
  In $\sigma_1$,  user $1$ provides consent for category $1$ at time-point $0$  with timestamp $10$; at time-point $1$ with timestamp $50$, the system uses user $1$'s data $3$ (with category $1$) and user $1$'s data $1$ (with category $2$).
  In $\sigma_2$, user $1$ requests deletion of data $1$, and then the system
  uses data $3$ with category $1$.
\end{example}
%
\subsection{Metric first-order temporal logic (MFOTL)}

MFOTL extends First-Order Logic (FOL) with 
temporal operators $\Next$
(`next'), $\Prev$ (`previous'), $\Eventually$ (`eventually') and $\Since$ (`since')~\cite{Pnueli1977}
equipped with intervals~\cite{Koymans1990} that encode 
constraints on timestamps.
We also include let bindings that introduce intermediate definitions, 
as well as SQL-style and user-defined aggregations. 
We thus follow Hublet et al.~\cite{Hublet2025}, except for the  
$\UUntil$ (`until') operator being replaced by the $\Eventually$ operator as a primitive, 
a deliberate restriction to the fragment we enforce.

Fix a set of variables $\Variables$, a set of \emph{function symbols} $\Functions$
with arity $\iota : \Functions \to \Nat$,
and a set \emph{aggregation operator symbols} $\Omega$
with arity $\iota' : \Omega \to \Nat^2$.
Each $f \in \Functions$ is interpreted by a fixed total function
$\hat f : \Domain^{\iota(f)} \to \Domain$ (e.g., arithmetic or string operations),
and each $\omega \in \Omega$ is interpreted  by a fixed total
function $\hat \omega$ mapping finite multisets over $\Domain^n$ to finite subsets of $\Domain^m$
where $\iota'(\omega) = (n,m)$. The syntax of MFOTL is defined as follows:
\begin{align*}
  t \Coloneqq~& x \mid c \mid f(\bar t) & 
  \varphi \Coloneqq~&\top \mid \mathit{e}(\bar{t}) \mid x = c \mid \neg \varphi \mid \varphi \land \varphi \mid \exists x.\ \varphi \mid \Next_I \varphi \mid \Prev_I \varphi \\
  &&& \mid\Eventually_I \varphi\mid \varphi \Since_I \varphi\mid \LLet{e(\bar x)}{\varphi}{\varphi} \mid \bar x \gets \omega(\bar t;\bar x)\;\varphi,
\end{align*}

\noindent where $e\in \Events$, $x \in \Variables$, $c \in \Domain$, and $I \in \Intervals$
ranges over non-empty intervals in $\Nat$.
We use standard abbreviations
$\bot \coloneqq \neg \top$,
$\varphi \vee \psi \coloneqq \neg (\neg \varphi \land \neg \psi)$,
$\varphi \rightarrow \psi \coloneqq \neg \varphi \vee \psi$,
$\varphi \leftrightarrow \psi \coloneqq (\varphi \rightarrow \psi) \wedge (\psi \rightarrow \varphi)$,
$\forall x.~\varphi \coloneqq \neg (\exists x.~\neg \varphi)$,
$\Once_I \varphi \coloneqq \top \Since_I \varphi$ (`once'), and 
$\Always_I \varphi \coloneqq \neg \Eventually_I \neg \varphi$ (`always'). 
and omit $[0,\infty)$ intervals from the temporal operators' subscript.
Furthermore, we write $\varphi[d/x]$ for the formula resulting from substituting 
the free variable $x$ with the constant $d$ in the formula $\varphi$.

\begin{example}\label{ex:formulae}
  Suppose the time unit is \emph{days}. Consider the following four formulae:

  \vspace{-1em}
  {\small
  \begin{align*}
    \varphi_{\mathsf{law}}  &= \Always (\forall c, d, u.\,\mathsf{use}(c, d, u) \rightarrow
        (\neg\,\mathsf{revoke}(u, c) \Since \mathsf{consent}(u, c)))\\
    \varphi_{\mathsf{del}}  &= \Always (\forall d, u.\,\mathsf{deletion\_request}(d, u) \rightarrow
        \Eventually_{[0,30]}\mathsf{delete}(d, u))\\
        \varphi_{\mathsf{let}}  &=  \text{\normalshape$\LLet{\mathsf{has\_consent}(u,c)}{\neg\mathsf{revoke}(u,c) \Since \mathsf{consent}(u,c)}{}$} \Always (\forall c, d, u.\,\mathsf{use}(c,d,u) \rightarrow \mathsf{has\_consent}(u,c))\\
    \varphi_{\mathsf{agg}}  &= \LetAgg{n}{\mathtt{CNT}}{d}{u}{\Once_{[0,10]} \mathsf{use}(c, d, u)}
  \end{align*}}
  \looseness=-1
  The first formula, $\varphi_{\mathsf{law}}$,
  formalizes the \emph{lawfulness of processing}: `whenever data $d$ with category $c$ belonging to user $u$ is processed, then $u$ must have consented to the use of category $c$ at some point in the past and not since revoked that consent.'
  The second formula, $\varphi_{\mathsf{del}}$, formalizes the GDPR's \emph{right to erasure}: `whenever a user $u$ requests the deletion of data $d$, then $d$ must be deleted within 30 days.'
  The third formula, $\varphi_{\mathsf{let}}$, illustrates \emph{let bindings}: it defines $\mathsf{has\_consent}(u,c)$ as holding when $u$ has consented to category $c$ and not revoked that consent, and uses this definition to restate $\varphi_{\mathsf{law}}$ more concisely.
  The fourth formula, $\varphi_{\mathsf{agg}}$, illustrates \emph{aggregations}: it counts, grouped by $u$, the number of data items of user $u$ used over the last 10 time units.
\end{example}

\looseness=-1
We write $\fv(\varphi)$ and $\fv(t)$ for the set of free variables of a formula $\varphi$ and a term $t$, respectively.
\begin{figure}[t]
  {\footnotesize
\renewcommand{\arraystretch}{0.92}
$
  \begin{array}{@{}l@{\;}c@{\;}l@{\ \ }|@{\ \ }l@{\;}c@{\;}l@{}}
      v,i \vDash e(\bar{t}) &\text{iff}& (e,\underv{\bar{t}}) \in D_i &
      v,i \vDash \lnot \varphi &\text{iff}& v,i \not\vDash \varphi\\
      v,i \vDash \varphi \land \psi &\text{iff}& v,i \vDash \varphi \text{  and } v,i \vDash \psi &
      v,i \vDash \exists x.\ \varphi&\text{iff}& v[x \mapsto d],i \vDash \varphi \text{  for some } d \in \Domain\\
      v,i \vDash \Next_I \varphi &\text{iff}& i + 1 < |\sigma|,\ v,i+1 \vDash \varphi \text{  and } \tau_{i+1} - \tau_i \in I &
      v,i \vDash \Prev_I \varphi&\text{iff}& i > 0,\ v,i-1 \vDash \varphi \text{, and } \tau_{i} - \tau_{i-1} \in I \\
      v,i \vDash \Eventually_I \varphi &\text{iff}& v,j \vDash \varphi \text { for some } j \geq i \text{  with } \tau_j - \tau_i \in I &
     v,i \vDash_\sigma \LLet{e(\bar x)}{\varphi}{\psi}&\text{iff}& v,i \vDash_{\sigma[e \mapsto \varphi]} \psi \\
     \multicolumn{6}{@{}l@{}}{v,i \vDash \varphi \Since_I \psi \;\;\text{iff}\;\; v,j\vDash \psi \text { for some } j \leq i \text{  with } \tau_i - \tau_j \in I\text{ and } v,k \vDash \varphi \text{  for all } j < k \leq i \hspace{5mm} \mid       v,i \vDash x = c\ \ \  \text{iff}\ v(x) = c} \\
     \multicolumn{6}{@{}l@{}}{v,i \vDash \bar y \gets \omega(\bar s;\bar g)\,\varphi \;\;\text{iff}\;\; v(\bar y) \in
       \hat \omega\,\mset{\sem{\bar s}_{v'} \mid v' \in \mathcal G}, \mathcal G = \{ v' \mid \mathrm{dom}(v') = \fv(\varphi),\ v'(\bar g) = v(\bar g),\ v',i \vDash \varphi \}, 0 < |\mathcal G| < \infty}
  \end{array}
  $}
  \caption{MFOTL semantics for a fixed trace $\sigma$. The free
  variables of an aggregation are $\bar g \cup \bar y$. $\mset{\cdots}$ is a multiset, 
  and $\sigma[e \mapsto \varphi]$ extends each
  $D_j$ of $\sigma$ with $\{e(\sem{\bar x}_{v'}) \mid \mathrm{dom}(v') \supseteq \fv(\varphi),\ v',j \vDash_\sigma \varphi\}$.}
  \label{fig:semantics}
  \vspace{-10pt}
\end{figure}
MFOTL's formal semantics is defined in \Cref{fig:semantics}.
Given a valuation $v: \Variables \rightharpoonup \Domain$ whose domain includes $\fv(\varphi)$,
we define the interpretation of terms as $\sem{x}_v=v(x)$ (for variables), $\sem{c}_v=c$
(for constants), and $\sem{f(\bar t)}_v=\hat f(\sem{\bar t}_v)$ (for function applications).
Each judgment $v, i \vDash_\sigma \varphi$ denotes that $\varphi$ is satisfied at
time-point $i$ of trace $\sigma$ under valuation $v$.
We omit $\sigma$ whenever it is clear from the context.


\subsection{Enforcers}
We consider a \emph{proactive} runtime enforcement model~\cite{Basin2024,Hublet2024}
where non-compliant
traces can be corrected in two ways: by modifying existing event databases by
causing or suppressing events (\emph{reactively});
or by inserting a single event database
before any timestamp increment (\emph{proactively}).
The latter feature allows enforcers
to satisfy real-time deadlines by inserting events.
We assume that a fixed subset
of all events can be caused at any time ($\Cau \subseteq \Events$),
while another fixed subset
can be suppressed ($\Sup \subseteq \Events$).
Formally, a \emph{proactive enforcer} is a quadruple
$\mathcal E =(\mathcal S, s_0, \mu, \nu)$ where
$\mathcal S$ is a set of states,
$s_0$ is an initial state,
$\mu : \mathcal S \times \Traces \times \Nat \times \Databases \to  \mathcal S \times \Databases_\Cau \times \Databases_\Sup$ updates the state and causes and suppresses
events when processing a database,
and $\nu : \mathcal S \times \Traces \times \Nat \to \mathcal S \times (\Databases_\Cau \cup \{\bot\})$ updates the state and optionally causes events (or does nothing) when
a timestamp is incremented. As in~\cite{Hublet2024,Hublet2025},
the enforcer $\mathcal E$ modifies a trace $\rho$ as
follows:

\begin{algorithm}
  \caption{Enforcement loop from Hublet et al.~\cite{Hublet2024}}
  \footnotesize
\begin{algorithmic}[1]
\Function{$\mathcal{E}$}{$\rho = (\tau_i, D_i)_i$}
  \State $s \gets s_0$;\ \ $\sigma \gets \varepsilon$
  \For{$i = 0, \ldots, |\rho|-1$}
    \State $(s, C, S) \gets \mu(s, \sigma, \tau_i, D_i)$;\ $\sigma \gets \sigma \cdot (\tau_i,\ (D_i \setminus S) \cup C)$ \Comment{reactive}
    \For{$t \in \{\tau_i, \ldots, \tau_{i+1}-1\}$ (none if $i = |\rho|-1$)}
      \State $(s, C') \gets \nu(s, \sigma, t)$;\ \textbf{if}\ $C' \neq \bot$\ \textbf{then}\ $\sigma \gets \sigma \cdot (t, C')$\ \textbf{end if} \Comment{proactive}
      \EndFor
  \EndFor
  \State \Return $\sigma$
\EndFunction
\end{algorithmic}
\end{algorithm}

The enforcement loop incrementally processes a trace $\rho$ of timestamped events $(\tau_i, D_i)$, maintaining a current state $s$ and an output trace $\sigma$. For each event, it first applies $\mu$ to update the state and transform the event's data (causing and suppressing event in $C$ and $S$), and then calls $\nu$ before any timestamp increment to proactively cause any additional events.

An enforcer $\mathcal E$ is \emph{sound} with respect to a closed formula
$\varphi$ if for any system trace $\sigma \in \Traces$, we have
$\emptyset, 0 \vDash_{\mathcal E(\sigma)} \varphi$, i.e., the modified trace
$\mathcal E(\sigma)$ complies with $\varphi$.

\section{EF: An Operational Enforcement Language}

\label{sec:ef}
\lstdefinelanguage{ef}{
  morekeywords={event,pyinit,fun,tfun,filter,let,table,lagged,agg,tableop,
                rule,trigger,add,remove,if,or,group_by,over,
                delay,next,window,section,fixpoint,once,sync,
                int,float,str,bool,true,false,
                SUM,AVG,MED,CNT,MIN,MAX,STD},
  sensitive=true,
  keywordstyle={\color{kwcolor}\bfseries},
  morecomment=[l]{//},
  morestring=[b]",
  literate={:=}{{$:=\ $}}1 {<-}{{$\leftarrow$}}1,
}
\lstset{basicstyle=\ttfamily\small,
        commentstyle=\itshape\color{gray!70},columns=fullflexible}

We now present EF, a domain-specific language (DSL)
that encodes the behavior of a proactive, first-order enforcer.
EF programs are trace transducers that cause and suppress events when
specific triggers hold, store temporal relations into tables, and perform
fixpoint computations. One key feature of EF programs is that their
per-time-point evaluation complexity can be easily bounded in terms of the size of
some of their internal tables. If these tables remain small, this
provides guarantees that the program's latency will not significantly grow with
trace size, a key feature for deployment in long-running systems.
In Havelund et al.'s terminology~\cite{Havelund2024}, EF is an
\emph{operational specification language} with a concrete executable semantics.
We present EF's syntax (Section~\ref{sec:syntax}) and its semantics (Section~\ref{sec:semantics})
and analyze the complexity of executing EF programs (Section~\ref{sec:complexity}).

\subsection{Syntax}

\label{sec:syntax}

We write $\mathit{id}$ for identifiers, $n$ for integer literals,
and $\seq{X}$
for a comma-separated list of $X$. Figure~\ref{fig:grammar} shows the full
grammar of EF programs. Figure~\ref{fig:example_ef} shows
the EF program enforcing the lawfulness policy $\varphi_\mathsf{law}$ of
Example~\ref{ex:formulae}, which we revisit in Section~\ref{sec:mfotl-to-ef}, 
generating it automatically from $\varphi_\mathsf{law}$.

An EF program is a sequence of \emph{declarations} followed by a sequence of
\emph{items}. The declarations define the enforcement signature and auxiliary code:
$\kw{event}\ \mathit{id}(\seq\tau)$ declares an event with typed
arguments; $\kw{pyinit}$ provides a Python preamble executed at start-up (imports and shared
globals); $\kw{fun}$ declares a Python function usable inside terms; and
$\kw{tfun}$ declares a custom aggregation operator mapping a list of rows to another.
The items define the enforcer's behavior:
\begin{itemize}
\item A \kw{table} stores a persistent finite relation.
       At every iteration, the interpreter removes the rows matched
        by its \kw{remove} clause and inserts those matched by its \kw{add}
        clause. A $\kwc{[}\kw{window}\ a\ b\kwc{]}$ annotation makes the table \emph{metric}:
        rows are active when they were inserted between $a$ and $b$ time units in the
        past ($b=\kwc{*}$ denotes $\infty$).
        A \kw{lagged} table is consumed in the \emph{next} iteration
        instead of the current one and cannot have a \kw{remove} clause.
      \item A \kw{let} defines a non-persistent relation within the current time-point.
        The \kw{filter} keyword signals that the \kw{let} expression
        may not be finite and can only be used as a boolean condition on top of another
        relation.
  \item An \kw{agg}~\kw{let} evaluates its \kw{over} clause, groups the results
        by \kw{group\_by}, and applies the operator
        $\mathit{op}$ ($\kw{SUM},\kw{AVG},\dots$) to each group.
  \item A \kw{rule} matches its \kw{trigger} to current events
        and, for each match, triggers an action. The action
        $\kwc{+}$ \emph{causes} an event, while $\kwc{-}$ \emph{suppresses} it. By
        default, the action is performed at the current time-point.
        With $\kwc{+}$, the $\kwc{[}\kw{delay}\ n\kwc{]}$ annotation triggers the event
        \emph{proactively} in $n$ time-units, while $\kwc{[}\kw{next}\ n\kwc{]}$ schedules it
        in $n$ time-points.
  \item A \kw{section} declaration groups all following rules (until the next \kw{section})
        into one evaluation block, run either \kw{once} or to a local \kw{fixpoint}, i.e., until
        the set of events obtained via suppressions and causations is stable.
        Every rule must belong to a section.
\end{itemize}

A \emph{clause} is a pair of
a disjunction $\pi$ of \emph{guards} $\kappa_1,\dots,\kappa_n$, $n \geq 0$,
with each guard being a conjunction of atoms, and a \emph{filter} $\kw{if}\ \varphi$. 
An atom is either a predicate $e(t_1,\dots,t_k)$ or an equality
$\mathit{id}\,\kwc{==}\,v$.
Without a guard, the clause is reduced to its filter.
If there is at least one guard, each free variable of the clause must be the 
argument of an atom of each guard. Guarded clauses can thus serve as queries over 
a relation: matching variable assignments are enumerated using the guards, and then
filtered.

\begin{figure}
\begin{lstlisting}[language=ef,aboveskip=0pt,belowskip=0pt,basicstyle=\tt\footnotesize]
event use(int, int, int);  // use(c, d, u)
event consent(int, int);   // consent(u, c)
event revoke(int, int);    // revoke(u, c)
event Sup_Exists0(int, int); event Sup_Exists2();
// [... standard library functions elided]
table Since1(c:int, u:int) := add {consent(u, c)} remove {revoke(u, c)};
let Exists2() := {use(c, d, u) if !Since1(c, u)};
section fixpoint;
rule +Sup_Exists2()     := trigger {if Exists2()};
rule +Sup_Exists0(c, u) := trigger {Sup_Exists2() & use(c, d, u) if !Since1(c, u)};
rule -use(c, d, u)      := trigger {Sup_Exists0(c, u) & use(c, d, u)};
\end{lstlisting}
\caption{EF program generated by \EnfFlash{} for $\varphi_\mathsf{law}$
(Example~\ref{ex:formulae}): table \texttt{Since1} implements the $\Since$
subformula, obligation events \texttt{Sup\_Exists}$x$ discharge existentials,
the last rule suppresses $\use(c,d,u)$ when consent fails.
\label{fig:example_ef}}

\smallskip
{\footnotesize
\[
\renewcommand{\arraystretch}{1.15}
\begin{array}{@{}r@{\;}c@{\;}l@{\qquad}r@{\;}c@{\;}l@{}}
\mathit{prog} &::=& \multicolumn{4}{@{}l@{}}{\mathit{evdecl}^*\;[\mathit{pyinit}]\;
                     \mathit{fun}^*\;\mathit{tfun}^*\;\mathit{item}^*} \\
\mathit{evdecl} &::=& \kw{event}\ \mathit{id}\kwc{(}\seq{\tau}\kwc{)}\,\kwc{;} &
\mathit{pyinit} &::=& \kw{pyinit}\ \kwc{\{}\ \textit{python}\ \kwc{\}} \\
\mathit{fun} &::=& \kw{fun}\ \mathit{id}\kwc{(}\seq{c}\kwc{)}\,\kwc{:}\,\tau\
                    \kwc{\{}\ \textit{python}\ \kwc{\}} &
\mathit{tfun} &::=& \kw{tfun}\ \mathit{id}\ \kwc{\{}\ \textit{python}\ \kwc{\}} \\
\mathit{item} &::=& \mathit{table}\mid\mathit{let}\mid\mathit{agg}\mid\mathit{rule}\mid\mathit{sec} &
\mathit{sec} &::=& \kw{section}\ (\kw{fixpoint}\mid\kw{once})\,\kwc{;} \\[4pt]
\mathit{table} &::=& \multicolumn{4}{@{}l@{}}{[\ell]\ [\kw{lagged}]\ \kw{table}\
                      \mathit{id}\kwc{(}\seq{c}\kwc{)}\ [\,\kwc{[}\kw{window}\ n\ b\kwc{]}\,]\
                      \kwc{:=}\ \kw{add}\ \kwc{\{}\,\mathit{cl}\,\kwc{\}}\
                      [\,\kw{remove}\ \kwc{\{}\,\mathit{cl}\,\kwc{\}}\,]\,\kwc{;}} \\
\mathit{let} &::=& \multicolumn{4}{@{}l@{}}{[\ell]\ [\kw{filter}]\ \kw{let}\
                    \mathit{id}\kwc{(}\seq{c}\kwc{)}\ \kwc{:=}\ \kwc{\{}\,\mathit{cl}\,\kwc{\}}\,\kwc{;}} \\
\mathit{agg} &::=& \multicolumn{4}{@{}l@{}}{[\ell]\ \kw{agg}\ \kw{let}\
                    \mathit{id}\kwc{(}\seq{c}\kwc{)}\ \kwc{:=}\ \mathit{id}\kwc{(}\seq{t}\kwc{)}\
                    \kw{group\_by}\ \seq{\mathit{id}}\
                    \kw{over}\ \kwc{\{}\,\mathit{cl}\,\kwc{\}}\,\kwc{;}} \\
\mathit{rule} &::=& \multicolumn{4}{@{}l@{}}{[\ell]\ \kw{rule}\ (\kwc{+}\mid\kwc{-})\
                     \mathit{id}\kwc{(}\seq{t}\kwc{)}\ [\,\kwc{[}\kw{delay}\ n\kwc{]}\,]\
                     [\,\kwc{[}\kw{next}\ n\kwc{]}\,]\
                     \kwc{:=}\ \kw{trigger}\ \kwc{\{}\,\mathit{cl}\,\kwc{\}}\,\kwc{;}} \\[4pt]
\mathit{cl} &::=& \pi\ [\kw{if}\ \varphi] \mid \kw{if}\ \varphi \quad
\pi ::= \kappa\ (\kw{or}\ \kappa)^* & \kappa &::=& \gamma\ (\kwc{\&}\ \gamma)^* \quad
\gamma ::= \mathit{id}\kwc{(}\seq{\mathit{id} \mid v}\kwc{)} \mid \mathit{id}\ \kwc{==}\ v \\
\varphi &::=& \kw{true}\mid\kw{false}\mid
               \mathit{id}\kwc{(}\seq{t}\kwc{)} \mid \varphi\,\kwc{\&}\,\varphi \mid
               \varphi\,\kwc{|}\,\varphi \mid \kwc{!}\varphi &
t &::=& \mathit{id} \mid \mathit{const} \mid \mathit{id}\kwc{(}\seq{t}\kwc{)} \\
\tau &::=& \kw{int}\mid\kw{float}\mid\kw{str}\mid\kw{bool} &
c &::=& \mathit{id}\,\kwc{:}\,\tau \quad \ell ::= \kwc{@}\mathit{id} \quad b ::= n \mid \kwc{*}
\end{array}
\]}
\vspace{-0.5em}
\caption{Grammar of EF programs\label{fig:grammar}}
\end{figure}

\subsection{Semantics}
\label{sec:semantics}

Given an EF program, the EF reference interpreter, described
in Algorithm~\ref{alg:saturate}, defines the behavior
of the $\mu$ and $\nu$ functions defined in Section~\ref{sec:preliminaries}.

The \emph{enforcer state} is a triple $(T, T^{\Next}, \Omega)$
where $T : \Events \rightharpoonup \mathcal{P}(\Nat \times \Domain^*)$ maps 
table names to (finite) timestamped relations, $T^{\Next} : \Events \rightharpoonup \mathcal{P}(\Nat \times \Domain^*)$ buffers the content of lagged tables, and $\Omega \in \mathcal{P}(\Events \times \Domain^* \times (\{ \mathtt{ts}, \mathtt{tp} \} \times \Nat))$ accumulates
proactive obligations created via \kw{delay} or \kw{next}, together with their deadlines.
Initially, all tables are empty and $\Omega = \emptyset$.

Each $\kw{fun}$ declaration associates a function $\hat{f}$ to the function
symbol $f$, and terms are evaluated as in MFOTL.
Under an interpretation $R: \Events \rightharpoonup \mathcal{P}(\Domain^*)$ of the
predicates, guards and clauses evaluate to sets of valuations as follows:
$\sem{p(\bar t)}_R=\{v \mid \mathrm{dom}(v) = \fv(\bar t),\ \sem{\bar t}_v \in R(p)\}$,
$\sem{x\,\kwc{==}\,c}_R=\{\{x\mapsto c\}\}$, and
$\sem{\kappa_1\,\kw{or}\,\dots\,\kw{or}\,\kappa_m\ \kw{if}\ \varphi}_R=\{v\in\textstyle\bigcup_{i}\Join_{\gamma\in\kappa_i}\sem{\gamma}_R\mid v\models_R\varphi\}$,
where $\Join$ denotes a join.

\newcommand{\proj}[2]{\overline{#1}^{#2}}
Algorithm~\ref{alg:saturate} defines the semantics of a program $P$ with
definitions $d_1,\dots,d_n$ of predicates $p_1,\dots,p_n$.
The \emph{working set} $\mathcal{S} = (T, \tau, D, C, S)$ consists of the tables, the
current timestamp, the input events, and the events caused and suppressed so far.
The tables contain \emph{timestamped rows} of the form $\tau \cdot r$ where $\tau$
is a timestamp and $r$ is a sequence of domain elements.
The operator $\proj{s}{I} =  \{ \bar a \mid \tau' \cdot \bar a \in s,\ \tau - \tau' \in I \}$
selects the rows of a timestamped relation whose timestamps are within the interval $I$.
The notation $v{\restriction}_{\bar x}$ denotes the row obtained by projecting $v$
on $\bar x$, i.e., $(v(x_1),\dots,v(x_k))$.

The reference interpreter consists of four functions.
\textsc{Eval} evaluates definitions on the working set.
\textsc{Interp} extends $R_0 = (D \setminus S) \cup C$ with the value of
each of the $p_i$, evaluating definitions in sequence using \textsc{Eval}.
\textsc{Update} applies a single rule on the current working set, updating $\Omega$.
\textsc{Saturate} applies rules section-wise (once for a
\kw{once} section, and until $(\Omega,C,S)$ reaches a fixed point for \kw{fixpoint} sections), updates the persistent tables, and returns a new state.
Throughout, $\mathcal S$ refers to the current values of $T$, $C$ and $S$.
The reactive function $\mu$ simply calls \textsc{Saturate} after retrieving the
\kw{next} obligations due at the current time-point.
The proactive $\nu$ function first computes the events to be caused
at the current timestamp or time-point, identifying obligations that hit their deadline.
Then, if and only there exist any such events, it calls the \textsc{Saturate} loop.

\begin{algorithm}[t]
\caption{EF interpreter for a program $P$ with predicates $p_1, \ldots, p_n$ and definitions $d_1, \ldots, d_n$}\label{alg:saturate}
\footnotesize
\newcommand{\cont}{\Statex\hspace*{\algorithmicindent}\hspace*{\algorithmicindent}}%
\newcommand{\brk}{\\\hspace*{\algorithmicindent}\hspace*{\algorithmicindent}}%
\newcommand{\gap}{\Statex\vspace{-0.6\baselineskip}}%
\algtext*{EndFunction}\algtext*{EndFor}\algtext*{EndIf}%
\begin{minipage}[t]{0.49\textwidth}
\begin{algorithmic}[1]
\Function{$\mu$}{$(T, T^{\Next}, \Omega), \sigma, \tau, D$}
\State $O \gets \{ e \mid (e, (\mathtt{tp}, |\sigma|)) \in \Omega \}$
\State $(T, C, S, \Omega) \gets {}$\Call{Saturate}{}
\cont $\big((T, \tau, D, O, \emptyset), T^{\Next}, \Omega, \sigma\big)$
\State \Return $\big((T, T^{\Next}, \Omega),\, C,\, S\big)$
\EndFunction
\gap
\Function{$\nu$}{$(T, T^{\Next}, \Omega), \sigma, \tau$}
  \State $C \gets \{ e \mid (e, (\mathtt{tp}, |\sigma|)) \in \Omega$
  \cont ${}\lor (e, (\mathtt{ts}, \tau)) \in \Omega \}$
  \State \textbf{if} $C = \emptyset$ \textbf{then} \Return $\big((T, T^{\Next}, \Omega),\, \bot\big)$
  \State $(T, C, S, \Omega) \gets {}$\Call{Saturate}{}
  \cont $\big((T, \tau, \emptyset, C, \emptyset), T^{\Next}, \Omega, \sigma\big)$
  \State \Return $\big((T, T^{\Next}, \Omega),\, C\big)$
\EndFunction
\gap
\Function{Saturate}{$\mathcal{S}=(T,\tau,D,C,S),T^{\Next},\Omega,\sigma$}
\ForAll{section $s$ of $P$ in order}
\Repeat
\ForAll{rule $r\in s$}
  \State $(\Omega,C,S) \gets \Call{Update}{r, \Omega, \mathcal{S}, \sigma}$
\EndFor
\Until{$(\Omega, C,S)$ unchanged \textbf{or} $s$ is \kw{once}}
\EndFor
\ForAll{\kw{table} $p_i(\bar x_i) := \kw{add}\ \{a_i\}\ \kw{remove}\ \{r_i\}$}
\State $T(p_i) \gets \big(T(p_i)\setminus \{\tau' \cdot v{\restriction}_{\bar x_i}\mid {}$
\cont $v\in\sem{r_i}_{\textsc{Interp}(\mathcal{S})}, \tau' \in \Nat\} \big)$
\cont ${}\cup\{\tau \cdot v{\restriction}_{\bar x_i}\mid v\in\sem{a_i}_{\textsc{Interp}(\mathcal{S})}\}$
\EndFor
\ForAll{\kw{lagged}\ \kw{table} $p_i(\bar x_i) := \kw{add}\ \{c_i\}$}
\State $T^{\Next}(p_i) \gets  \{\tau \cdot v{\restriction}_{\bar x_i}\mid v\in\sem{c_i}_{\textsc{Interp}(\mathcal{S})}\}$
\EndFor
\ForAll{\kw{lagged}\ \kw{table} $p_i$}
\State $T(p_i) \gets T^{\Next}(p_i)$
\EndFor
\State \Return $(T,C,S,\Omega)$
\EndFunction
\algstore{efsem}
\end{algorithmic}
\end{minipage}\hfill
\begin{minipage}[t]{0.49\textwidth}
\begin{algorithmic}[1]
\algrestore{efsem}
\Function{Interp}{$\mathcal{S}=(T,\tau,D,C,S)$}
\State $R_0 \gets (D \setminus S) \cup C$
\For{$i = 1, \dots, n$}
  \State $R_i \gets R_{i-1}[p_i \mapsto \Call{Eval}{d_i, R_{i-1}, T, \tau}]$
\EndFor
\State \Return $R_n$
\EndFunction
\gap
\Function{Eval}{$d, R, T, \tau$}
\If{$d = \kw{let}\ p(\bar x) := \{c\}$}
  \State \Return $\{ v{\restriction}_{\bar x} \mid v\in\sem{c}_R \}$
\ElsIf{$d = \kw{table}\ p(\bar x)\ [\kw{window}\ n\ b] := {}$\brk
  $\kw{add}\,\{a\}\,\kw{remove}\,\{r\}$}
  \State \Return $\big(\proj{T(p)}{[n,b]}\setminus \{v{\restriction}_{\bar x}\mid v\in\sem{r}_R\}\big)$
  \cont ${}\cup(n = 0\ ?\ \{v{\restriction}_{\bar x}\mid v\in\sem{a}_R\} : \emptyset)$
\ElsIf{$d = \kw{lagged}\ \kw{table}\ p(\bar x)\ [\kw{window}\ n\ b] := {}$\brk
  $\kw{add}\,\{c\}$}
  \State \Return $\proj{T(p)}{[n,b]}$
\ElsIf{$d = \kw{agg}\ \kw{let}\ p(\bar g,\bar o):=g(\bar s)$\brk
  $\kw{group\_by}\ \bar g\ \kw{over}\,\{c\}$}
  \State \Return $\big\{ \bar v(\bar g) \cdot\bar w \mid \bar v \in \sem{c}_R,$
  \cont $\bar w\in\hat g\,\langle\!\langle\,\sem{\bar s}_v \mid v\in\sem{c}_R, v(\bar g) = \bar v(\bar g)\,\rangle\!\rangle \big\}$
\EndIf
\EndFunction
\gap
\Function{Update}{$r, \Omega, \mathcal{S}=(T,\tau,D,C,S), \sigma$}
\Statex \hspace*{\algorithmicindent}\textit{// $r = \kw{rule}\ \alpha\,p(\bar t)\ \dots := \kw{trigger}\ \{c\}$}
\State $A \gets \{\sem{\bar t}_v \mid v\in\sem{c}_{\textsc{Interp}(\mathcal{S})}\}$
\If{$r$ has $\kwc{[}\kw{delay}\ N\kwc{]}$}
  \State \Return $\big(\Omega \cup \{ (p(\bar a), (\mathtt{ts}, \tau + N))$
  \cont ${}\mid \bar a \in A \},\ C,\ S\big)$
\ElsIf{$r$ has $\kwc{[}\kw{next}\ N\kwc{]}$}
  \State \Return $\big(\Omega \cup \{ (p(\bar a), (\mathtt{tp}, |\sigma| + N))$
  \cont ${}\mid \bar a \in A \},\ C,\ S\big)$
\ElsIf{$\alpha = \kwc{+}$}
  \State \Return $(\Omega,\ C \cup \{ p(\bar a) \mid \bar a \in A\},\ S)$
\Else
  \State \Return $(\Omega,\ C,\ S \cup \{ p(\bar a) \mid \bar a \in A \})$
\EndIf
\EndFunction
\end{algorithmic}
\end{minipage}
\end{algorithm}

\begin{example}\label{ex:execution}
We run the program of Figure~\ref{fig:example_ef} on the trace $\sigma_1$ of
Example~\ref{ex:traces}.
The initial state is $(T,T^{\Next},\Omega)=(\{\mathtt{Since1}\mapsto\emptyset\},\emptyset,\emptyset)$. Since the program has no \kw{delay} or \kw{next} rules, $\Omega$ remains
empty and $\nu$ always returns $\bot$.

At time-point 0, we have $R_0=\{\mathtt{consent}\mapsto\{(1,1)\}\}$, then
$\textsc{Interp}(\mathcal{S})(\mathtt{Since1})=\emptyset\cup\{(1,1)\}$ and
$\textsc{Interp}(\mathcal{S})(\mathtt{Exists2})=\emptyset$ as there is no \texttt{use}.
No trigger matches, hence the fixpoint stops after one pass with $C=S=\emptyset$, and
$\mu$ outputs $D$. The \textsc{Saturate} function commits $T(\mathtt{Since1})=\{10\cdot(1,1)\}$.

At time-point 1, we still have $\textsc{Interp}(\mathcal{S})(\mathtt{Since1})=\{(1,1)\}$.
Since user 1 never consented to category 2,
$\textsc{Interp}(\mathcal{S})(\mathtt{Exists2})=\{()\}$.
The first iteration of the fixpoint proceeds as follows:
\begin{enumerate}
\item The rule $\kwc{+}\mathtt{Sup\_Exists2}()$ is activated, yielding $C=\{\mathtt{Sup\_Exists2}()\}$;
\item The trigger of $\kwc{+}\mathtt{Sup\_Exists0}(c,u)$ joins $\mathtt{Sup\_Exists2}()$
  with \texttt{use}  and filters by $\kwc{!}\mathtt{Since1}(c,u)$,
  obtaining the valuation
  $v=\{c\mapsto 2, d\mapsto 1, u\mapsto 1\}$;
  hence $\mathtt{Sup\_Exists0}(2,1)$ is added to $C$;
\item The rule $\kwc{-}\mathtt{use}(c,d,u)$ matches $\mathtt{Sup\_Exists0}(2,1)$ with
  $\mathtt{use}(2,1,1)$, yielding $S=\{\mathtt{use}(2,1,1)\}$.
\end{enumerate}
In the second iteration of the fixpoint,
$\mathtt{use}(2,1,1)\notin R_0=(D\setminus S)\cup C$ since
$\mathtt{use}(2,1,1) \in S$, and hence
$\textsc{Interp}(\mathcal{S})(\mathtt{Exists2})=\emptyset$.
The set $(\Omega,C,S)$ is unchanged and the fixpoint stops. The output is
$(D\setminus S)\cup C=\{\mathtt{use}(1,3,1),\mathtt{Sup\_Exists2}(),\mathtt{Sup\_Exists0}(2,1)\}$; the unlawful use is suppressed.
\end{example}

\subsection{Complexity}
\label{sec:complexity}

We bound the runtime of executing a program $P$ in terms of three parameters:
the largest number of guards in any single conjunction $\kappa$ of any clause of $P$,
which we denote $w$;
the largest number of tuples in any relation $\mathcal S(p)$ occurring in at least
one guard, which we denote $N$;
and the total number of fixpoint iterations performed, which we denote $K$.
We make the following assumptions that let
us isolate relational work:
\begin{enumerate}
\item[(A1)] the original trace carries $O(1)$ events per time-point;
\item[(A2)] each scalar function $\hat f$, comparison, and aggregation step is
  $O(1)$, membership tests, insertion, and deletion within a relation is amortized
  $O(1)$ with hash maps, and any $\hat \omega$ on $m$ input rows runs in $O(m)$,
  returning $O(m)$ rows.
\end{enumerate}
Under (A1) the input-event relations have size $O(1)$, so $N$ is governed by the
tables.
We prove:

\begin{lemma}\label{lem:filter}
The boolean $\varphi$, $v \models_{\mathcal{S}} \varphi$ can be computed in $O(1)$.
\end{lemma}
\begin{proof}
Since $\varphi$ has constant size, so it suffices that each atom is $O(1)$. A comparison
$t_1\bowtie t_2$ evaluates its constant-size terms under $v$ in $O(1)$ by (A2). A
predicate a let-bound atom $p(\bar t)$ involves the membership test $\sem{\bar t}_v\in\mathcal S(p)$,
which costs $O(1)$ using hashsets. For a table,
membership in $\mathcal W(T(p))$ can be checked in
$O(1)$ by maintaining two separate subtables (timestamps
in $\tau-[a,b]$ and in $\tau-[0,a)$) and moving elements
between subtables on update.
\end{proof}

\begin{lemma}\label{lem:clause}
Given $R_{i-1}$, computing $\sem{c}_{R_{i-1}}$ for any clause $c$ with all identifiers in the domain of $R_{i-1}$ takes $O(N^{w})$ time.
\end{lemma}
\begin{proof}
  Matching a conjunction joins up to guards, each ranging over a relation with at most $N$
  tuples. The join enumerates at most $N^{w}$ valuations, each requiring an
$O(1)$ filter by Lemma~\ref{lem:filter}.
\end{proof}

\begin{lemma}\label{lem:rule}
Evaluating one rule of $P$ takes $O(N^{w})$ time.
\end{lemma}
\begin{proof}
With memoization, each relation the trigger depends on is
computed at most once, in $O(N^{w})$ by Lemma~\ref{lem:clause}. Matching the
trigger and instantiating its up to $N^{w}$ heads are $O(N^{w})$.
\end{proof}

\begin{theorem}\label{thm:tp}
One call to $\mu$ or $\nu$ takes $O\!\left(K\cdot N^{w}\right)$ time.
\end{theorem}
\begin{proof}
Each of the $K$ iterations is one pass over the $O(1)$ rules of a
section, each rule $O(N^{w})$ by Lemma~\ref{lem:rule}, hence $O(N^{w})$ per
iteration. Window updates touch each tuple
$O(1)$ times, committing the layer recomputes  constantly many table
relations ($O(N^{w})$). Each emitted
obligation---at most $O(K \cdot N^w)$ per time-point---is discharged once by $\nu$ in $O(1)$ by evaluating the 
validation filter.
\end{proof}

\begin{theorem}
A trace of $L$ time-points is processed in
$O\!\left(L\cdot K_{\max}\cdot N_{\max}^{w}\right)$, where $K_{\max}$ and
$N_{\max}$ are the largest iteration count and relation size over all time-points.
\end{theorem}
\begin{proof}
  Sum the complexity of Theorem~\ref{thm:tp} over the $L$ time-points.
\end{proof}

Note that if no tables are used in guards, such as in
Example~\ref{fig:example_ef}, then $N_{\max} = O(1)$ by 
(A1). If $P$ can be shown to always reach
its fixed points in a bounded number 
of steps, this ensures that the total complexity is
$O(L)$, an  amortized constant per-step complexity. 
This is particularly desirable in long-running systems
where large enforcer states can accumulate and penalize
evaluation latency.

\section{From Enforceable MFOTL to EF}
\label{sec:mfotl-to-ef}

In this section, we present an algorithm that generates EF programs
enforcing formulae in an expressive fragment of  MFOTL. 
Our generation algorithm consists of four components: a normalization
procedure that places all temporal subformulae of a formula behind 
let bindings (\ref{sec:lnf}); rewrite rules that extract guards from
subformulae (\ref{sec:guards}); rewrite rules
that generate enforceable clauses from a formula (\ref{sec:rewrite}); 
and a dependency analysis that decides enforceability of a conjunction of
clauses and fixes their evaluation order (\ref{sec:edg}).

\subsection{Let-normal form}
\label{sec:lnf}

In this section, we work with MFOTL formulae rewritten into 
\emph{let-normal form} with the shape
\begin{align*}
    &\LLet{e_1(\bar x_1)}{\varphi_1\!}{}\!\ldots \LLet{e_m(\bar x_m)}{\varphi_m\!}{} \Always \chi_1 \land \ldots \land \Always \chi_n,
\end{align*}
where each $\varphi_i$ is one of $\psi$, $\Prev_I \psi$, $\psi_l \Since_I \psi_r$, or $\LetAgg{\bar y}{\omega}{\bar t}{\bar g}{\psi}$ and $\psi$ is $\exists y_1,\dots,y_k. \chi$, with
$\chi$ an MFOTL formula without any quantifiers (except over subformulae with future operators),
$\Prev$ and $\Since_I$  operators, or aggregation. 
With formulae in let-normal form, quantification, past-temporal, and aggregation are applied as the last step of evaluating a let binding.
One shows:
\begin{theorem}
    For any formula $\Always \varphi$ in MFOTL, there exists an 
    equivalent $\psi$ in let-normal form.\leanref{let_normal_form}
\end{theorem}

\subsection{Guard extraction}
\label{sec:guards}

Consider a quantifier-free FOL formula $\Phi$.
Next, we show how to derive a disjunction of guards $\pi = \bigvee_i \bigwedge_{j} \gamma_{ij}$
and a filter $\varphi$  such that 
$\pi \land \varphi \equiv \Phi$ and every guard $\kappa$ in $\pi$ \emph{binds} every variable
$x$ of a given set $X$, i.e., $\kappa$ contains an atom $p(\bar t)$ with $x \in \bar t$ or an atom $x = c$.

Given a set $m$ of enumerable events that can be used as guards,
we define a transformation $m \vdash \Phi \rightsquigarrow^+_X (\pi,\varphi')$
that finds a guard for every $x \in X$ in $\Phi$ and collects the resulting
\emph{joint guard} of $X$ in $\pi$,
preserving $\Phi \equiv \pi \land \varphi$.
The rules defining $\rightsquigarrow^+_X$ and its dual $\rightsquigarrow^-_X$
are shown in Figure~\ref{fig:jointguards}.

\newcommand{\guards}{\mathrm{Guards}}

Rules \textsc{Pred} and \textsc{Eq} insert a new guard from an enumerable event and
an equation, respectively. \textsc{Vac} provides the base case for $\top$.
Rule \textsc{Neg} exchanges $\rightsquigarrow^+_X$ for $\rightsquigarrow^-_X$. 
Rules \textsc{And}$^+$ and \textsc{And}$^-$ extract guards from conjuncts.
The former relies on the equation
$\pi \land \varphi_j \equiv \pi' \land \varphi'_j \Longrightarrow \pi \land \bigwedge_i \varphi_i \equiv \pi' \land \bigwedge_{i \neq j} \varphi_i \land \varphi'_j$,
which enables a guard to be extracted from any conjunct.
The latter uses the dual equation
$\forall i.~\pi \to \varphi_i \equiv \pi'_i \to \varphi'_i \Longrightarrow \pi \to \bigwedge_i \varphi_i \equiv \bigvee_i \pi'_i \to \bigwedge_i (\pi'_i \to \varphi'_i)$.
We write
$\guards^m_X(\Phi) = (\pi, \varphi)$ if $m \vdash \Phi \rightsquigarrow^+_X (\pi, \varphi)$
(for a fixed choice among several derivations),
and $\guards^m_X(\Phi) = \bot$ if no such derivation exists and show:

\begin{figure}[t]
\footnotesize
\renewcommand{\MathparLineskip}{\lineskiplimit=.6\baselineskip\lineskip=.6\baselineskip plus .2\baselineskip}
\begin{mathpar}
\inferrule*[left=\textsc{None}]{ }{m \vdash \Phi \rightsquigarrow^p_\emptyset (\{\top\}, \Phi)}
\and
\inferrule*[left=\textsc{Vac}]{ }{m \vdash \top \rightsquigarrow^-_X (\emptyset, \top)}
\and
\inferrule*[left=\textsc{Pred}]{p \in m \\ X \subseteq \bar t}
  {m \vdash p(\bar t) \rightsquigarrow^+_X (\{p(\bar t)\}, \top)}
\and
\inferrule*[left=\textsc{Eq}]{X \subseteq \{x\}}
  {m \vdash x = c \rightsquigarrow^+_X (\{x = c\}, \top)}
\and
\inferrule*[left=\textsc{And}$^+$]
  {m \vdash \varphi \rightsquigarrow^+_{X_1} (\pi_1, \varphi') \\
   m \vdash \psi \rightsquigarrow^+_{X_2} (\pi_2, \psi') \\ X \subseteq X_1 \cup X_2}
  {m \vdash \varphi \land \psi \rightsquigarrow^+_X
     (\{\kappa_1 \land \kappa_2 \mid \kappa_1 \in \pi_1, \kappa_2 \in \pi_2\}, \varphi' \land \psi')}
\and
\inferrule*[left=\textsc{And}$^-$]
  {m \vdash \varphi \rightsquigarrow^-_{X} (\pi_1, \varphi') \\
   m \vdash \psi \rightsquigarrow^-_{X} (\pi_2, \psi')}
  {m \vdash \varphi \land \psi \rightsquigarrow^-_X
     (\pi_1 \cup \pi_2, (\pi_1 \to \varphi') \land (\pi_2 \to \psi'))}
\and
\inferrule*[left=\textsc{Neg}]{m \vdash \varphi \rightsquigarrow^{-p}_X (\pi, \varphi')}
  {m \vdash \neg\varphi \rightsquigarrow^{p}_X (\pi, \neg\varphi')}
\end{mathpar}
\caption{Joint guard extraction}
\label{fig:jointguards}
\vspace{-1em}
\end{figure}

\begin{lemma}
  If $\guards^m_X(\Phi) = (\pi, \varphi)$, then every $\kappa \in \pi$ binds every $x \in X$, and
  $\bigvee_{\kappa \in \pi}\kappa \land \varphi \equiv \Phi$.\leanref{Lemma_4_2}
\end{lemma}

\noindent
For a trigger $(\pi, \psi)$ and a variable $x$, we write $m \vdash (\pi, \psi) \rightsquigarrow^p_x (\pi', \psi')$
if either every $\kappa \in \pi$ binds $x$ and $(\pi', \psi') = (\pi, \psi)$, or
$m \vdash \psi \rightsquigarrow^p_{\{x\}} (\pi_0, \psi')$ and
$\pi' = \{\kappa \land \kappa_0 \mid \kappa \in \pi, \kappa_0 \in \pi_0\}$.

\subsection{Enforcement rewriting}
\label{sec:rewrite}

Our rewrite rules turns a formula into a set of \emph{enforcement clauses}.
A clause $\theta\Rightarrow\varepsilon$ pairs a trigger $\theta=(\pi,\varphi)$
with an \emph{effect} $\varepsilon$
of the form $e(\bar t)$, $\neg e(\bar t)$, $\Eventually_I e(\bar t)$,
or $\Next \ldots \Next e(\bar t)$. A clause denotes the rule
$\Always\,\forall{\bar x}.\;\theta\to\varepsilon$ with
$\bar x =\fv(\theta)\cup \fv(\varepsilon)$: for any
valuation satisfying the trigger, the effects are applied.

The judgement $\Gamma\vdash\varphi\hookrightarrow^\Cau \mathcal C$ produces a set of sets of
clauses $\mathcal C = \{ C_k \}_k$ such that enforcing \emph{any}
alternative $C_k$ enforces $\varphi$.
Dually, the judgement $\Gamma\vdash\varphi\hookrightarrow^\Sup \mathcal C$
is such that enforcing any alternative enforces $\neg \varphi$.
The environment $\Gamma$ maps let-bound
event names to a triple of booleans $(g, c, s)$
that indicate whether the corresponding event is guarded, causable,
and suppressable, respectively.

Let $\mathcal C[f]$ denote
$\{\{ f(\pi,\psi,\varepsilon) \mid (\pi, \psi) \Rightarrow \varepsilon \in C \} \mid C \in \mathcal C\}$
and $\mathcal C_1 \otimes \mathcal C_2$ denote $\{C_1 \cup C_2 \mid (C_1,C_2) \in \mathcal C_1 \times \mathcal C_2 \}$.
The enforcement rewrite rules are shown in Figure~\ref{fig:rewrite}.
We write $m_\Gamma$ for the set of enumerable events: the base events and the let-bound $e$
with $\Gamma(e)_1 = \top$.

\begin{figure*}[t]
\footnotesize
\renewcommand{\MathparLineskip}{\lineskiplimit=.4\baselineskip\lineskip=.4\baselineskip plus .1\baselineskip}
\begin{mathpar}
\inferrule*[left=$\top^\Cau$]{ }{\Gamma\vdash\top\hookrightarrow^\Cau \{\emptyset\}}
\and
\inferrule*[left=\textsc{Ev}$^\Cau$]{e \in \Cau}
  {\Gamma\vdash e(\bar t)\hookrightarrow^\Cau \{\{ (\top, \top) \Rightarrow e(\bar t)\} \}}
\and
\inferrule*[left=\textsc{Ev}$^\Sup$]{e \in \Sup}
  {\Gamma\vdash e(\bar t)\hookrightarrow^\Sup \{\{ (e(\bar t), \top) \Rightarrow \neg e(\bar t)\} \}}
\and
\inferrule*[left=\textsc{Let}$^\Cau$]{\Gamma(e)_2 = \top}
  {\Gamma\vdash e(\bar t)\hookrightarrow^\Cau \{\{ (\top, \top) \Rightarrow \mathtt{Cau}\_e(\bar t)\} \}}
\and
\inferrule*[left=\textsc{Let}$^\Sup$]{\Gamma(e)_3 = \top}
  {\Gamma\vdash e(\bar t)\hookrightarrow^\Sup \{\{ (e(\bar t),\top) \Rightarrow \mathtt{Sup}\_e(\bar t)\} \}}
\and
\inferrule*[left=\textsc{Neg}]{\Gamma\vdash\varphi\hookrightarrow^{\bar t}\mathcal C}
  {\Gamma\vdash\neg\varphi\hookrightarrow^t\mathcal C}
\and
\inferrule*[left=\textsc{And}$^\Sup$]
  {\Gamma\vdash\varphi_j\hookrightarrow^\Sup\mathcal {\mathcal C}_j \\ \varphi_i~\text{present for}~i \neq j}
  {\Gamma\vdash\textstyle\bigwedge_i\varphi_i\hookrightarrow^\Sup \mathcal C_j[\pi,\psi,\varepsilon \mapsto (\pi, \psi \land \textstyle\bigwedge_{i\neq j} \varphi_i) \Rightarrow \varepsilon]}
\and  
\inferrule*[left=\textsc{And}$^\Cau$]
  {\forall i.~\Gamma\vdash\varphi_i\hookrightarrow^\Cau\mathcal {\mathcal C}_i}
  {\Gamma\vdash\textstyle\bigwedge_i\varphi_i\hookrightarrow^\Cau \bigotimes_i \mathcal C_i}
\and
\inferrule*[left=\textsc{Ex}$^\Cau$]{ \Gamma \vdash \varphi \hookrightarrow^\Cau \mathcal C }
           {\Gamma\vdash\exists x.\,\varphi\hookrightarrow^\Cau \mathcal C[\pi,\psi,\varepsilon \mapsto (\pi[0/x], \psi[0/x]) \Rightarrow \varepsilon[0/x]]}
\and
\inferrule*[left=\textsc{Ex}$^\Sup$]
  {\displaystyle\Gamma\vdash\varphi\hookrightarrow^\Sup\mathcal C \qquad
  \displaystyle\mathcal C' = \{ C \in \mathcal C \mid \forall (\pi,\psi) \Rightarrow \varepsilon \in C.~ \exists \pi', \psi'.~(\pi, \psi) \rightsquigarrow^+_x (\pi', \psi')\} }
  {\Gamma\vdash\exists x.\,\varphi\hookrightarrow^\Sup \{ \{ (\pi', \psi') \Rightarrow \varepsilon \mid (\pi, \psi) \Rightarrow \varepsilon \in C, (\pi, \psi) \rightsquigarrow^+_x (\pi',\psi') \} \mid C \in \mathcal C' \}}
\and
\inferrule*[left=\textsc{Fut}$_{\Eventually}^\Cau$]
  {\displaystyle\Gamma\vdash\varphi\hookrightarrow^\Cau\mathcal C\qquad
      \displaystyle\mathcal C' = \{ C \in \mathcal{C} \mid C \neq \emptyset \land \forall (\pi, \psi) \Rightarrow \varepsilon \in C.~\exists p, \bar t.~\varepsilon=p(\bar t) \land \pi = \psi = \top \} \qquad 1 \le b < \infty}
  {\Gamma\vdash\Eventually_{[a,b]}\varphi\hookrightarrow^\Cau \mathcal C'[\pi,\psi,\varepsilon \mapsto (\top,\top) \Rightarrow  \Eventually_{[b, b]} \varepsilon]}
\and
\hspace{-2em}\inferrule*[left=\textsc{Fut}$_{\Next}^\Cau$]
  {\displaystyle\Gamma\vdash\varphi\hookrightarrow^\Cau\mathcal C \qquad 1 \le b\qquad
      \displaystyle\mathcal C' = \{ C \in \mathcal{C} \mid C \neq \emptyset \land \forall (\pi, \psi) \Rightarrow \varepsilon \in C.~\exists p, \bar t.~\varepsilon=p(\bar t) \land \pi = \psi = \top \}}
  {\Gamma\vdash\Next_{[0,b]}\varphi\hookrightarrow^\Cau \mathcal C'[\pi,\psi,\varepsilon \mapsto (\top,\top) \Rightarrow \Next \varepsilon]}
\and
\hspace{-2em}\inferrule*[left=\textsc{Fut}$_{\Next^n}^\Cau$]
  {\displaystyle\Gamma\vdash\varphi\hookrightarrow^\Cau\mathcal C\qquad 1 \le n\qquad
      \displaystyle\mathcal C' = \{ C \in \mathcal{C} \mid \forall (\pi, \psi) \Rightarrow \varepsilon \in C.~\exists p, \bar t.~\varepsilon=p(\bar t) \land \pi = \psi = \top \}}
  {\Gamma\vdash\underbrace{\Next\ldots\Next}_{n}\varphi\hookrightarrow^\Cau \mathcal C'[\pi,\psi,\varepsilon \mapsto (\top,\top) \Rightarrow \Next\ldots\Next \varepsilon]}
\and
\end{mathpar}
\vspace{-2em}
\caption{Enforcement rewrite rules}
\label{fig:rewrite}
\end{figure*}

The two \textsc{Ev} rules are the only rules implementing immediate causation or
suppression.
If the signature declares $e$ causable ($e\in\Cau$), then $e$ is enforced by an
unconditional clause that causes $e(\bar t)$; if it declares $e$ suppressable
($e\in\Sup$), the clause fires exactly when $e(\bar t)$ is present and suppresses it.
The \textsc{Let} rules assert a fresh
\emph{obligation event} $\mathtt{Cau}\_e$ (resp.\ $\mathtt{Sup}\_e$), of the same arity as $e$, that will be
consumed by clauses that realize $e$'s body. The \textsc{Neg} rule
flips the target formula's polarity: causing $\neg\varphi$ is suppressing $\varphi$
($\bar\Cau = \Sup$, $\bar\Sup = \Cau$).
To cause a conjunction (\textsc{And}$^\Cau$), every conjunct must be caused, so the clause
sets are combined by the product $\bigotimes$; to suppress it
(\textsc{And}$^\Sup$) it suffices to falsify one conjunct $\varphi_j$, with the
other conjuncts added to the filter. To cause
$\exists x.\,\varphi$ (\textsc{Ex}$^\Cau$), the
rule substitutes a canonical value $0$ for $x$. To
suppress it (\textsc{Ex}$^\Sup$), every witness must be falsified, hence
$x$ must be guarded.

The temporal rules implement deferred causation. The \textsc{Fut}$_{\Eventually}$ and
\textsc{Fut}$_{\Next}$ rules apply only when the sub-derivation has reduced the body
to a single causable event; the effect is scheduled at $b$ resp. after $n$
time points. Since
$\nu$ only inserts a time-point when something is due, the intermediate
time-points of a chain of $\Next$ operators need not exist within a bounded
interval; hence bounded intervals are only supported for a single $\Next$
(\textsc{Fut}$_{\Next}$), whereas chains must be unbounded (\textsc{Fut}$_{\Next^n}$).
Deadlines are at least one time unit (resp.\ time-point) ahead, so that obligations
created by $\nu$ are discharged by a later call. 

\subsection{Generating the candidate clause sets}
\label{sec:gen}

We write
$\mathrm{gate}^\Cau_p(\mathcal C)$ (resp. $\mathrm{gate}^\Sup_p$)
for the set of sets of clauses obtained by conjoining the atom $\mathtt{Cau}\_p(\bar x)$
(resp. $\mathtt{Sup}\_p$)
to every guard $\pi$ in $\mathcal C$, and $\Gamma\vdash\varphi\hookrightarrow^\alpha$ for the set of all
$C \in \mathcal C$ with $\Gamma\vdash\varphi\hookrightarrow^\alpha\mathcal C$.
Rewriting is described in Algorithm~\ref{alg:gen}. The algorithm
computes a let-normal form $(\mathcal L, \varphi^\flat)$ of $\Always\varphi$ as in Section~\ref{sec:lnf}, with
$\mathcal L$ its lets and $\varphi^\flat = \chi_1 \land \dots \land \chi_n$ the conjunction of its $\Always$-bodies.
It
first processes every let in order, generating the triple
$\Gamma(p)=(g,c,s)$ and causation/suppression clauses
$\mathcal C^\Cau(p),\mathcal C^\Sup(p)$,
and then rewrites
$\varphi^\flat$. The output is a \emph{set of candidate clause sets}; to
yield a valid enforcer, a clause set will have to pass the tests in
Section~\ref{sec:edg}.

\begin{algorithm}[t]
\caption{Generation of candidate clause sets}
\label{alg:gen}
\begin{algorithmic}[1]
  \footnotesize
\Function{Generate}{$\Always\varphi$}
  \State $(\mathcal L,\varphi^\flat)\gets\textsc{LetNormalForm}(\Always\varphi)$;\quad
         $\Gamma\gets\emptyset$
  \State \textbf{for all} $\big(p(\bar x):=\psi\big)\in\mathcal L$ \textbf{do} $\Gamma\gets\textsc{TypeLet}(\Gamma,\,p,\bar x,\psi)$
  \State \Return $\bigcup_{C \in \mathcal C.\Gamma\vdash\varphi^\flat\hookrightarrow^\Cau\mathcal C}  \textsc{Realizations}(\Gamma,C)$
\EndFunction
\vspace{1em}
\Function{TypeLet}{$\Gamma,\,p,\bar x,\psi$}
\State $\chi\gets\textsc{StripExists}(\psi)$; $m\gets m_\Gamma$; $\mathcal C^\Cau(p), \mathcal C^\Sup(p) \gets \emptyset$
\Switch{$\chi$}
    \Case{$\Once_{[a,b]}\varphi$}
       \State $\theta\gets\guards^m_{\fv(\varphi)}(\varphi)$;\;
              \textbf{if} $\theta=\bot$ \textbf{then reject end if};\;
              $\mathcal C^\Cau(p)\gets
              a = 0\ ?\ 
              \mathrm{gate}^\Cau_p(\Gamma\vdash\varphi\hookrightarrow^\Cau) : \emptyset$
    \EndCase
    \Case{$\Prev\varphi \;\mid\; \mathtt{Agg}(\dots,\varphi)$}
       \State $\theta\gets\guards^{m}_{\fv(\varphi)}(\varphi)$;\;
              \textbf{if} $\theta=\bot$ \textbf{then reject end if}
    \EndCase
    \Case{$\varphi_l\Since_{[a,b]}\varphi_r$}
       \State $\theta_l\gets\guards^{m}_{\bar x\cup\fv(\chi)}(\neg\varphi_l)$;  $\theta_r\gets\guards^{m}_{\bar x\cup\fv(\chi)}(\varphi_r)$
       \State \textbf{if} $\theta_l=\bot \lor \theta_r=\bot$ \textbf{then reject end if}
       \State \textbf{if} $a=0$ \textbf{then} $\mathcal C^\Cau(p)\gets \mathrm{gate}^\Cau_p(\Gamma\vdash\varphi_r\hookrightarrow^\Cau)$;\ 
       $\mathcal C^\Sup(p)\gets
       \mathrm{gate}^\Sup_p\big((\Gamma\vdash\varphi_l\hookrightarrow^\Sup)
       \otimes(\Gamma\vdash\varphi_r\hookrightarrow^\Sup)\big)$
       \State \hspace{3.1em} \textbf{else} $\mathcal C^\Sup(p)\gets
       \mathrm{gate}^\Sup_p(\Gamma\vdash\varphi_l\hookrightarrow^\Sup)$ \textbf{end if}
    \EndCase
    \Case{otherwise}
       \State $\theta\gets\guards^{m}_{\bar x\cup\fv(\chi)}(\chi)$;\; \textbf{if} $\theta=\bot$ \textbf{then return} $\Gamma[p\mapsto(\bot,\bot,\bot)]$ \textbf{end if}
       \State \textbf{if} $\chi$ is present \textbf{then} $\mathcal C^\Cau(p)\gets\mathrm{gate}^\Cau_p(\Gamma\vdash\psi\hookrightarrow^\Cau)$;\; $\mathcal C^\Sup(p)\gets\mathrm{gate}^\Sup_p(\Gamma\vdash\psi\hookrightarrow^\Sup)$ \textbf{end if}
    \EndCase
  \EndSwitch
  \State \Return $\Gamma\big[p\mapsto(\top,\;
         \mathcal C^\Cau(p)\neq\emptyset,\;
         \mathcal C^\Sup(p)\neq\emptyset)\big]$
\EndFunction
\end{algorithmic}
\end{algorithm}
\textsc{TypeLet} performs a case split on the let body's outermost operator.
Each case decides \emph{guardedness} $g$ and computes the
two clause families by rewriting the relevant operand and gating it with $p$'s
obligation event. The boolean $c$ (resp. $s$) can be computed by checking
the emptiness of the set of causation (resp. suppression) clauses.
Which clauses are generated depends on the topmost operator of the let.
A $\Once$ or $\Since$
is causable only when its interval admits the present ($a=0$) by causing
the (right) operand; a $\Since$ is suppressable when either $a>0$ and
its left operand is suppressable or both its operands are.
A $\Prev$ or aggregation is neither causable or suppressable.
Insufficient guarding is fatal for a temporal or
aggregation let. For any other let, it 
downgrades $p$ to \emph{filter-only}, making it usable as a
boolean condition but not caused or suppressed.

Once $\Gamma$ is complete, \textsc{Generate} rewrites $\varphi^\flat$ at $\Cau$,
yielding a set of candidate clause sets. \textsc{Realizations}$(\Gamma, C)$ is the set of the least clause sets $R \supseteq C$ such that,
for some choice of $f(p) \in \mathcal C^\Cau(p)$ and $g(p) \in \mathcal C^\Sup(p)$ for every let $p$,
$f(p) \subseteq R$ (resp.\ $g(p) \subseteq R$) whenever the effect of a clause of $R$ is on
$\mathtt{Cau\_}p$ (resp.\ $\mathtt{Sup\_}p$), possibly under $\Eventually$ or $\Next$. The
resulting candidate clause sets serve as input to the dependency analysis
in the next section, which filters clause sets yielding a terminating enforcer.

\subsection{Dependency analysis}
\label{sec:edg}

Not all candidate clause sets generated in the previous step yield a sound enforcer.
For example, the set
$C_1=\{ (\top, \top) \Rightarrow A(0), (\top, \top) \Rightarrow \neg A(0) \}$ is
unsatisfiable,
while $C_2 = \{ (A(x), \top) \Rightarrow A(x+1) \}$ generates
an infinite number of events and never terminates.
We state two checks that together guarantee that a set of clauses
can be finitely enforced. The first rules out events that are simultaneously caused
and suppressed, as in $C_1$;
the second rules out non-termination, as in $C_2$.

\paragraph{Cause/suppress conflicts}
Consider the \emph{Event Dependency Graph} (EDG) whose vertices are the base and obligation events
(let atoms are decomposed into the events defining them),
and there is an edge $e \to e'$ labeled with $c_k^{a}$ when
$e$ occurs in the trigger of some clause $k$ and $e'$ occurs in its effect with polarity
$a\in\{\Cau,\Sup\}$ (deferred effects have polarity $\Cau$)  (Figure~\ref{fig:edg}).
Its strongly connected components (SCCs) are
mutually-influencing event sets. For each $e$ that one clause
with trigger $\theta_\Cau$ causes and another clause with trigger $\theta_\Sup$
suppresses, the pair $(\theta_\Cau,\theta_\Sup)$ is sent
to an SMT solver, which tests the unsatisfiability of $\theta_\Cau \land \theta_\Sup$
by only taking into account events from SCCs
\emph{(strictly) before} $e$ in the topological order of SSCs.
If one of the effects is deferred (\kw{delay} or \kw{next}), the two
triggers are evaluated at different time-points, and no event is shared.
Any \textsc{Sat} result rejects the clause set.
This rules out $C_1$  above but allows, e.g.,
$\{ (\top, B(0)) \Rightarrow A(0), (\top, \neg B(0)) \Rightarrow \neg A(0) \}$,
which is equivalent to $\Always A(0) \leftrightarrow B(0)$.

\begin{figure}[t]
\vspace{-5mm}
\centering
\begin{minipage}[b]{0.6\textwidth}
\centering
\scalebox{0.5}{  
\begin{tikzpicture}[>=latex,line join=bevel,   p/.style ={draw, circle, inner sep=1.5pt, minimum size=8mm}]

\begin{scope}
  \pgfsetstrokecolor{black}
  \definecolor{strokecol}{rgb}{0.0,0.0,0.0}
  \pgfsetstrokecolor{strokecol}
  \draw [draw=red!65] (100.47bp,8.0bp) -- (100.47bp,159.0bp) -- (406.69bp,159.0bp) -- (406.69bp,8.0bp) -- cycle;
  \draw[red!65] (382.58bp,20.5bp) node {SCC 0};
\end{scope}
  \node[p] (n1) at (137.85bp,75.0bp) {$\mathtt{Sup\_Exists0}$};
  \node[p] (n5) at (253.58bp,50.0bp) {$\mathtt{use}$};
  \node[p] (n2) at (369.31bp,100.0bp) {$\mathtt{Sup\_Exists2}$};
  \node[p] (n3) at (19.737bp,95.0bp) {$\mathtt{consent}$};
  \node[p] (n4) at (19.737bp,54.0bp) {$\mathtt{revoke}$};
  \draw [->, bend right] (n1) to node[below] {$c_2^{\mathbb S}$} (n5);
  \draw [->, bend right] (n2) to node[below] {$c_1^{\mathbb C}$} (n1);
  \draw [->, bend right] (n5) to node[below] {$c_1^{\mathbb C}$} (n1);
  \draw [->] (n5) to node[above] {$c_0^{\mathbb C}$} (n2);
  \draw [->, loop above] (n5) to node {$c_2^{\mathbb S}$} (n5);
  \draw [->] (n3) to node[below] {$c_1^{\mathbb C}$} (n1);
  \draw [->] (n4) to node[below] {$c_1^{\mathbb C}$} (n1);
  \draw [->, bend left] (n3) to node[below] {$c_0^{\mathbb C}$} (n2);
  \draw [->, bend right] (n4) to node[above, pos=0.3] {$c_0^{\mathbb C}$} (n2);
\end{tikzpicture}}
\caption{EDG of the program in Figure~\ref{fig:example_ef}. In its SCC, $\mathtt{use}$
is only suppressed ($c_2^{\Sup}$) and $\mathtt{Sup\_Exists0/2}$ only caused
($c_{0,1}^{\Cau}$), so the check succeeds.}
\label{fig:edg}
\end{minipage}\hfill
\begin{minipage}[b]{0.37\textwidth}
\centering
\scalebox{0.75}{\begin{tikzpicture}[>={Stealth[round]}, node distance=20mm,
   p/.style ={draw, circle, inner sep=1.5pt, minimum size=8mm},
   ok/.style={draw=green!55!black, thick},
   bad/.style={draw=red!65, thick},
   lbl/.style={font=\scriptsize, align=center}]
  \node[p, ok]  (u) {$\mathtt{use.1}$};
  \node[p, bad, below=10mm of u] (a) {$\mathtt{A.0}$};
  \draw[->, green!55!black] (u) to[out=-40,in=40,looseness=8]
        node[lbl, right]{stable\\ $d\!\mapsto\!d$} (u);
  \draw[->, red!65, dashed] (a) to[out=-40,in=40,looseness=8]
        node[lbl, right]{non-stable\\ $x\!\mapsto\!x{+}1$} (a);
\end{tikzpicture}}
\caption{Data-flow graph: the self-loop on \texttt{use.1} (running example)
is stable; the one on \texttt{A.0} ($A(x)\Rightarrow A(x{+}1)$) is not.}
\label{fig:flow}
\end{minipage}
\end{figure}

\paragraph{Termination}
Termination is decided on a distinct \emph{Data-Flow Graph} (DFG) whose nodes are
event arguments $e.i$ and that contains an edge $e.i \to e'.j$ for any clause
whose trigger mentions $e(\dots,t_i=x,\dots)$ and
whose effect is $e'(\dots,t'_j,\dots)$ with $x$ occurring in $t'_j$, and for any let
$e'(\bar x) := \varphi$ such that $\varphi$ mentions $e(\dots,t_i=x,\dots)$ and $x = x_j$ or,
if $\varphi$ is an aggregation, $x_j$ is one of its results.
An edge is \emph{stable} iff none of its endpoints is the result of an aggregation and
all function symbols in $t'_j$ (if any) are stable under a fixed preorder $\preceq$ on $\Domain$, i.e., $\hat f(\bar a) \preceq a_k$ for some $k$
and all $\bar a$ and $f$. Applying stable functions
repeatedly to finitely many values yields finitely many values~\cite{Hublet2025}.
Enforcement terminates if no non-stable edge lies on a cycle.
As a result, causing a fresh value on every iteration
(e.g.\ $(A(x),\top)\Rightarrow A(x+1)$) is rejected, while, e.g.,
$(\mathtt{use}(d),\top)\Rightarrow\neg\mathtt{use}(d)$ is accepted
because no new domain values are generated by that clause (Figure~\ref{fig:flow}).

\subsection{Compilation}
\label{sec:algo}

Consider a closed formula $\Always\varphi$ with let-normal form $(\mathcal L, \varphi^\flat)$, the
environment $\Gamma$ and the guards $\theta$ computed for each let by \textsc{TypeLet}
(Section~\ref{sec:gen}), and a candidate clause set $R$ that fulfills
both checks in Section~\ref{sec:edg}. Fix a topological order $\prec$ on the
condensed SCCs of its EDG (sources first), which also induces an order on clauses.
The compilation algorithm in Algorithm~\ref{alg:emit} turns $(\mathcal L,\Gamma,R)$
into an EF program that enforces $\Always \varphi$.

\begin{algorithm}[t]
  \caption{Compilation algorithm}
  \footnotesize
\label{alg:emit}
\begin{algorithmic}[1]
\Function{Compile}{$\mathcal L,\,\Gamma,\,R,\,\prec$}
  \State $P\gets$ declaration of all base/obligation events in $R$
  \ForAll{$p\in\Gamma$}
     \State append to $P$ the item realizing $p$'s body: \kw{table}
     (past operator), \kw{agg}, or \kw{filter\ let}
     (non-guarded present) or \kw{let}
  \EndFor
  \ForAll{$(\pi,\psi)\Rightarrow\varepsilon\in R$, by SCC in order $\prec$}
     \State open a \kw{section} \kw{fixpoint} for the SCC if not already open
     \State append the rule $\;[\,\pm\,]\,\hat\varepsilon\;:=\;
            \kw{trigger}\ \{\,\pi\ \kw{if}\ \psi\,\}$, where $\hat\varepsilon$
            carries a $[\kw{delay}\ n]$ (resp. $[\kw{next}\ n]$) annotation if
            $\varepsilon$ is $\Eventually_{[n,n]} p(...)$ (resp. $\Next \cdot \Next p(...)$ with $n$ $\Next$)
  \EndFor
  \State \Return $P$
\EndFunction
\end{algorithmic}
\end{algorithm}

Each clause becomes an EF rule; a caused effect
yields a $\kwc{+}$ rule, a suppressed effect a $\kwc{-}$ rule, and a deferred
effect a proactive $\kw{delay}$/$\kw{next}$ causation. Each let $p(\bar x) := \psi$ of $\mathcal L$ becomes an item built from the guards
$\theta$ of \textsc{TypeLet}: $\Once_{[a,b]}\varphi$ becomes
$\kw{table}\ p(\bar x)\ \kwc{[}\kw{window}\ a\ b\kwc{]} := \kw{add}\ \{\theta\}$, $\Prev_{[a,b]}\varphi$ the
same \kw{lagged} table, $\varphi_l \Since_{[a,b]} \varphi_r$ becomes
$\kw{table}\ p(\bar x)\ \kwc{[}\kw{window}\ a\ b\kwc{]} := \kw{add}\ \{\theta_r\}\ \kw{remove}\ \{\theta_l\}$,
an aggregation an \kw{agg} \kw{let} over $\{\theta\}$, and a present body
$\kw{let}\ p(\bar x) := \{\theta\}$, or $\kw{filter}\ \kw{let}\ p(\bar x) := \{\kw{if}\ \chi\}$ if $p$ is
filter-only. A trigger $(\pi, \psi)$ becomes $\pi\ \kw{if}\ \psi$, or $\kw{if}\ \psi$ if $\pi = \{\top\}$;
compilation fails if a trigger or filter has no EF counterpart. The SCC order $\prec$ provides the EF sections.

\begin{example}\label{ex:compile}
Running \EnfFlash{} on $\varphi_\mathsf{law}$ (Example~\ref{ex:formulae}) yields the
program of Figure~\ref{fig:example_ef}. Let-normalization binds the since
subformula $\neg\revoke(u,c)\Since\consent(u,c)$ to a let, which \textsc{TypeLet}
compiles to the windowless table $\mathtt{Since1}$: its right operand $\consent$
populates the \emph{\kw{add}} clause and its left operand $\revoke$ the \emph{\kw{remove}}
clause. The whole quantifier block is kept together, so the negated body
$\exists c,d,u.\,\use(c,d,u)\wedge\neg\mathtt{Since1}(c,u)$ is bound to the 
Boolean let $\mathtt{Exists2}$. Suppressing it yields three clauses in one
\emph{\kw{fixpoint}} section: $\mathtt{+Sup\_Exists0/2}$ raises the obligation event
whenever a violation is present, and $\mathtt{-use}$ then suppresses each offending
$\use(c,d,u)$ for which consent is missing or has been revoked. The resulting EDG
(Figure~\ref{fig:edg}) has a three-node SCC whose only suppressed event is
$\use$, so the conflict check passes; the only data-flow cycle is the stable
self-loop $\use(c,d,u)\Rightarrow\neg\use(c,d,u)$ (Figure~\ref{fig:flow}, left),
so termination holds.
\end{example}

\begin{theorem}[Compilation correctness]\label{thm:compile}
Let $R \in \textsc{Generate}(\Always\varphi)$ be a candidate clause set passing the two checks
of Section~\ref{sec:edg} and $P=\textsc{Compile}(\mathcal L,\Gamma,R,\prec)$.
Then $P$ is a sound enforcer for $\Always\varphi$.\leanref{Theorem_4_3}
\end{theorem}

\section{Evaluation}
\label{sec:eval}

We implement the compilation of Section~\ref{sec:mfotl-to-ef} and the EF
interpreter of Section~\ref{sec:ef} in a tool called \EnfFlash{}. The compiler
(7k LOC of OCaml) parses an MFOTL specification, runs let-normalization, guard
extraction, enforcement rewriting, and the dependency analysis, discharges the
cause/suppress checks with Z3, and emits an EF program. The interpreter
(4k LOC of Rust) loads the EF program, computes a bound on its complexity as in
Section~\ref{sec:ef}, 
and exposes the reactive and proactive
entry points $\mu$ and $\nu$. We use Rust's \texttt{pyo3} Python bindings to evaluate
user-defined functions. \EnfFlash{} is open-source and available \href{\repourl}{online}.
Our evaluation answers four research questions:
\begin{description}
   \item[RQ1] (Expressiveness) Does \EnfFlash{} support the specifications
     of existing benchmarks and real applications?
  \item[RQ2] (Latency) How does \EnfFlash{}'s latency
        compare to state-of-the-art first-order enforcers?
  \item[RQ3] (Overhead in web applications) Is the end-to-end enforcement overhead
        small enough for user-facing operations in web applications?
  \item[RQ4] (Comparison with permission engines) How does \EnfFlash{} compare to
        a state-of-the-art non-temporal permission engine (Cedar) and its
        temporal extension (Dogwood) on access control,
        accounting for the cost of maintaining temporal state?
\end{description}

\subsection{Experimental setup}
\label{sec:eval-setup}

All experiments run under Ubuntu 24.04.3
on an Intel Core Ultra 7 165U with 12 cores and 32GB of RAM.
We run two types of experiments:
tool-only benchmarks, where a trace from an existing benchmark is replayed
to enforcement tools; and application benchmarks,
where those tools serve as Policy Decision Points (PDP) for web applications~\cite{Hublet2025b}.

In tool-only benchmarks,
latency is measured as wall-clock time for the processing of single time points.
We report per-time-point average and maximum latency for each run.

In application benchmarks, we measure full round-trip page latency
using Python's \texttt{requests} library
against a local server with the \href{https://wiki.archlinux.org/title/CPU\_frequency\_scaling}{`performance' CPU governor} activated.

\begin{figure}[!t]
\centering
\includegraphics[width=0.455\textwidth]{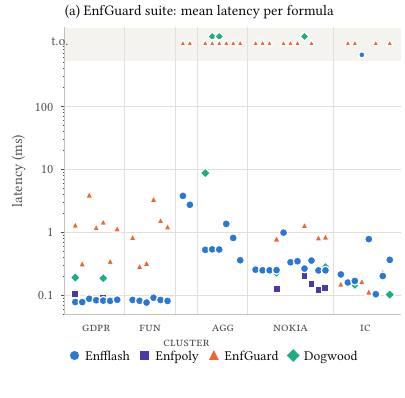}\hfill
\includegraphics[width=0.455\textwidth]{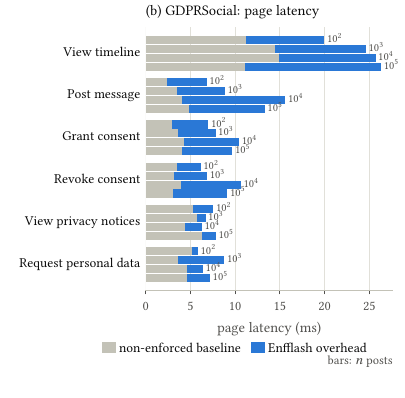}\\
\includegraphics[width=0.455\textwidth]{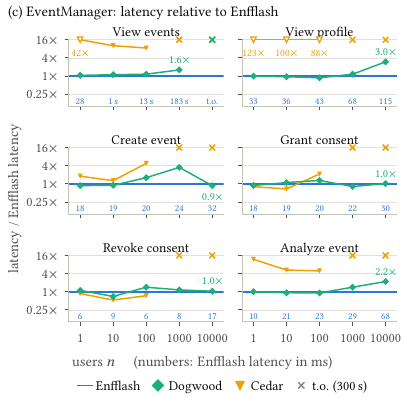}\hfill
\includegraphics[width=0.455\textwidth]{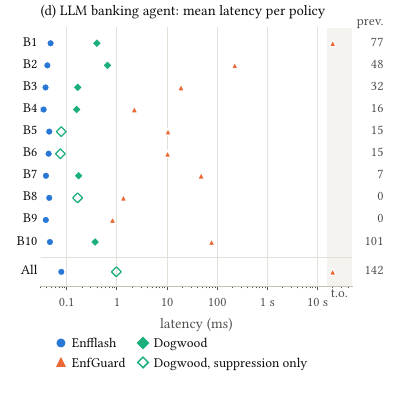}
\caption{Evaluation results\ifextended{} (Tables~\ref{tab:micro}--\ref{tab:app3})\fi;
  in (d), prev.\ is the number of attacks prevented out of 143.\ifextended\else{}
  Details in~\cite{extended}.\fi}
\label{fig:eval}
\vspace{-1em}
\end{figure}

\subsection{Tool-only: the \EnfGuard{} suite (RQ1,2)}
\label{sec:eval-micro}

We first evaluate \EnfFlash{} on \EnfGuard{}'s six enforcement benchmarks,
spanning data protection (\textsc{gdpr}, \textsc{fun}),
aggregation-heavy policies (\textsc{agg}),
fault clustering (\textsc{cluster}), infrastructure monitoring (\textsc{ic}),
and database updates in a large system (\textsc{nokia}).
We use all formulae in the benchmark with the largest available log.
For each, we compile the specification and replay~\cite{Krstic2020}
the provided traces in real time, measuring latency against \EnfFlash{}, \EnfGuard{}~\cite{Hublet2025},
\EnfPoly{}~\cite{Hublet2022}, \ifextended \MonPoly{}~\cite{Basin2011}, \else\fi and Dogwood~\cite{dogwood}
on identical inputs, with 3 repetitions.
We selected \EnfGuard{} and \EnfPoly{}
as the two fastest available MFOTL enforcers. \ifextended \MonPoly{} is an optimized
MFOTL \emph{monitor}, i.e.~a tool that identifies but does not correct policy
violations. Although \MonPoly{} is therefore not a direct competitor to enforcement tools, it
provides a challenging state-of-the-art latency baseline for a slightly simpler problem. \else\fi
Dogwood is a recent extension of Cedar with past-time temporal conditions over
an event history. It supports suppression but no causation, and its temporal conditions
are past-only, with bounded windows, no disjunction, and aggregations limited to
counts and sums of integers.

\paragraph*{Results (Figure~\ref{fig:eval}a)} 
Dogwood supports only 11 of the 40 formulae: other formulae require causing events 
(the \textsf{deletion} formulae and \textsc{cluster}), future operators (\textsc{nokia}'s
\textsf{del} and \textsf{ins} formulae), user-defined functions (\textsc{fun}), or
aggregations other than sums and counts (\textsc{agg}'s p4--p6).
On these, Dogwood is within 12\% of \EnfFlash{} on 4, 2.4--17$\times$ slower on 3,
3.6$\times$ faster on 1, and times out on 3.
\EnfFlash{} supports all policies in the benchmarks, with only the \textsc{nokia}
formulae require rewriting to avoid nesting past and future operators.
It processes 39 of the 40 within the timeout while EnfGuard processes 21.
On 17 of the 20 processed by both, \EnfFlash{}'s average latency is 3--44$\times$
lower than \EnfGuard{}'s, which is faster on 2 \textsc{ic} formulae (1.4$\times$, 7$\times$).
\EnfPoly{}, which supports only 7 formulae, is 1.3--2.3$\times$ faster than \EnfFlash{} on 5 \textsc{nokia} formulae.
\ifextended \EnfFlash{} is within 4$\times$ of \MonPoly{} on the 27 formulae processed by both
(1.1--1.8$\times$ faster on \textsc{gdpr}, up to 3.8$\times$ slower on \textsc{agg} and
\textsc{nokia}); \MonPoly{} times out on 4 \textsc{ic} formulae.\fi
Average (resp. max) latency is 0.08--3.8\,ms (resp. 1--98\,ms), below 1\,ms on 36 formulae.

\subsection{Application: GDPR-compliant social network (RQ1,3)}
\label{sec:eval-app1}

Our first deployment instruments \textsc{GDPRSocial},
a microblogging app implemented
in ca. 2.5k LOC of Python/Django that runs against an extensive GDPR specification
encoded as a 400-line MFOTL formula. The GDPR specification is intrinsically
\emph{temporal}, with both statefulness (consent and revocation thereof must be
persisted) and proactivity (personal data must be accessed, rectified, or deleted in
a timely manner when users require it).
\EnfFlash{} acts as an enforcement backend invoked through the Instrlib~\cite{Hublet2025b} library.
We measure the latency of a set of representative application views: accessing a user's
timeline, posting a message, giving or revoking consent for personalized advertisement,
viewing privacy notices, and requesting access to one's data. 
Measurements are performed for $n \in \{10^2, 10^3, 10^4, 10^5\}$ posts by
$n/100$ users over 10 repetitions after 5 warm-up calls.
As a baseline, we  perform measurements on a copy of the app
stripped of all enforcement.
Note that due to the temporal nature of the policy and the much higher latency
of \EnfGuard{}, no systematic experiments could be performed with competing tools.
Dogwood cannot express the specification, as it requires causing events
(e.g., deleting or providing personal data upon request).

\paragraph*{Results (Figure~\ref{fig:eval}b)}
\EnfFlash{} supports the large input formula.
All views load in 6--27\,ms, with enforcement
inducing an overhead of 1--15\,ms with respect to the non-enforced baseline.
This loading time would typically be perceived as instantaneous by
users~\cite[p. 135]{Nielsen1993}.
In particular, runtime latency does not grow with $n$, and hence enforcement
scales well with database size.

\subsection{Application: event-management platform (RQ1,3,4)}
\label{sec:eval-app2}

\looseness=-1
Our second deployment is \textsc{EventManager}, an event-management platform
whose access-control policy mixes role-based rules (e.g., only admins may add moderators) with stateful constraints (e.g., only event owners may promote an attendant of an event
to manager of the event). This application was proposed 
by Basin et al.~\cite{BasinGKNO23} and implemented in both Cedar~\cite{Cutler2024} and $\nu$ActionGUI~\cite{KrsticNB24} AC mechanisms.
Because most of this policy is expressible without temporal operators when the roles are fixed,
it admits a direct encoding in these two platforms, giving us two non-temporal baselines.
However, the fact that these roles can change over time means that an extensive dynamic context
must be passed to the non-temporal evaluator at every request, potentially increasing the latency.
We straightforwardly translated the platform's access control policy into ca. 150 lines of MFOTL
by replacing all the predicates of the form $\mathsf{hasRole}(\bar{x})$ by the
$\neg \mathsf{leftRole}(\bar{x}) \Since \mathsf{becomeRole}(\bar{x})$ formula.
As a temporal baseline, we translated the MFOTL formula into Dogwood,
which, unlike Cedar, also expresses the policy's consent rules: each
$\neg \mathsf{leftRole}(\bar{x}) \Since \mathsf{becomeRole}(\bar{x})$ becomes a
\texttt{since} condition, and each disjunction a separate \texttt{permit}.
The Dogwood version of the application is otherwise identical to the \EnfFlash{} one,
with the same enforcement points and the same result cache.
Measurements are performed for $n \in \{1, 10, \dots, 10^4\}$ users and $10n$ events;
requests time out after 300\,s (Cedar beyond $n = 100$, `view events' beyond $n = 1000$).

\paragraph*{Results (Figure~\ref{fig:eval}c)} Our \EnfFlash{} implementation is faster
than the Cedar baseline on 4 of 6 views  and within 7\,ms on the 2 consent views (where Cedar
is up to 1.9$\times$ faster), with 8--42$\times$ and 88--123$\times$ speedups on the
`view events' and `view profile' pages, which query permissions about many entities.
Dogwood is on par with \EnfFlash{} up to $n = 100$, but its latency grows faster with the
history: for $n \geq 1000$, it is 1.6--3.5$\times$ slower on 4 views
(e.g., 3.0$\times$ on `view profile' at $n = 10^4$).

\subsection{Tool-only: LLM banking agent (RQ1,2)}
\label{sec:eval-app3}

Our third deployment enforces security policies on the tool calls of an LLM agent.
We use the `banking' suite of AgentDojo~\cite{Debenedetti2024}, an established
benchmark for agentic security. In this suite, an agent operates a bank account
using 11 tools (e.g., sending money, changing a password, reading files).
The suite contains 16 benign and 144 attacked ground-truth runs, which we replay
after translating all tool calls, tool outputs, and user prompts into log events.
Each call is submitted to \EnfFlash{} for approval before it executes,
and suppressed calls are not executed. 

We specify ten policies B1--B10 in MFOTL. For instance, B1 restricts payees to
accounts mentioned by the user, paid before, or occurring in a file named by the user;
B2 forbids payment subjects containing account data not supplied by the user; and
B3 caps the amount sent to a payee within 24 hours at \$10,000.
All policies are first-order; they also  use unbounded past operators (B1,2,7,10),
aggregations over metric windows (B3,4), and causation (B5,6,8,9).
In contrast, existing runtime guardrails for LLM agents check either single tool
calls~\cite{Wang2026} or propositional temporal properties~\cite{Xiao2026}.
Dogwood cannot cause events: we translate B1--B8 and B10 into Dogwood (for B5, B6,
and B8, only their suppression part), emulating unbounded past operators by
100-year windows as in Section~\ref{sec:eval-micro}. We record the event log of
each policy and replay it through \EnfFlash{}, \EnfGuard{}, and Dogwood one time point
at a time, measuring latency as in Section~\ref{sec:eval-micro} (timeout 600\,s);
on every log, the three tools suppress the same tool calls.\looseness=-1

\paragraph*{Results (Figure~\ref{fig:eval}d)}
\EnfFlash{} supports all ten policies and their conjunction. None of the policies
suppresses a call in a benign run. Without enforcement, 143 of the 144 attacks
succeed; the conjunction prevents 142 of them, with B1 alone preventing 77
attacks. Average latency is 0.03--0.05\,ms per time point for individual policies,
and 0.08\,ms for the conjunction of B1--B10. Dogwood can fully enforce 6 policies
and partially enforce the suppression part of 3 additional ones (hollow markers in
Figure~\ref{fig:eval}d).
An additional experiment on a large
log of 23,320 shows that latency grows sublinearly with history side
(52 to 64\,$\mu$s from the first to the last fifth of the run).
\EnfGuard{} is 20--5,570$\times$ slower, while Dogwood is 1.8--16$\times$ slower and
12$\times$ slower on the conjunction.

\subsection{Summary}
\label{sec:eval-threats}

\EnfFlash{} supports all policies in the EnfGuard suite and the three application benchmarks,
showing that the supported fragment of MFOTL is expressive enough to support a
broad range of use cases (RQ1).
Our tool achieves average per-event latency of 0.08--3.8\,ms, 3--44$\times$
lower than \EnfGuard{} on 17 of the 20 formulae on which \EnfGuard{} does not time out\ifextended,
and within 4$\times$ of \MonPoly{}, a monitor solving a substantially simpler problem (RQ2)\fi. On our LLM banking agent benchmark, \EnfFlash{} enforces ten security policies
with 0.03--0.08\,ms average latency, 20--5,570$\times$ faster than \EnfGuard{} and
1.8--16$\times$ faster than Dogwood,
and prevents 142 of 143 successful prompt-injection attacks without suppressing
any benign tool call (RQ1,2). In a real agentic deployment, \EnfFlash{}'s sub-millisecond
latency would be dominated by the costs of inference.
In the \textsc{GDPRSocial} application, enforcement adds 1--15\,ms
of overhead despite the policy's complexity (400 LOC), and this overhead does not significantly
change with database size, keeping all views near-instantaneous (RQ3).
\EnfFlash{} matches or significantly improves on Cedar's performance in our
\textsc{EventManager} application, up to 3.5$\times$ faster than Dogwood on long histories
and up to 123$\times$ faster than Cedar's full reload-and-evaluate pipeline in dynamic
contexts (RQ4).

\section{Related work}

\label{sec:rw}

\begin{table}[tb]
\setlength{\tabcolsep}{3pt}
\centering
\caption{Temporal enforcers (\checkmark: supported, $\sim$: partially, --: not supported)}
\label{tab:rw}
\scriptsize
\setlength{\tabcolsep}{1.8pt}
\renewcommand{\arraystretch}{0.9}%
\begin{tabular}[t]{lcccccccc}
\toprule
\emph{Non-agentic} & \rot{First-order\\data} & \rot{Metric\\time} & \rot{Aggregations} & \rot{Suppression}
 & \rot{Causation} & \rot{Proactive\\deadlines} & \rot{Sub-ms\\latency} & \rot{Mechanized\\proofs} \\
\midrule
Timed enforcers~\cite{Pinisetty2015,Renard2017,FalconeP19,Falcone2021} & -- & \checkmark & -- & \checkmark & $\sim$ & $\sim$ & -- & -- \\
Hall\'e and Villemaire~\cite{HalleV12}                                 & \checkmark & -- & -- & \checkmark & -- & -- & -- & -- \\
\EnfPoly{}~\cite{Hublet2022}                                              & \checkmark & \checkmark & -- & \checkmark & \checkmark & -- & \checkmark & -- \\
Aceto et al.~\cite{Aceto2023}                                          & \checkmark & -- & -- & \checkmark & \checkmark & -- & -- & -- \\
Basin et al.~\cite{Basin2024}                                          & -- & \checkmark & -- & \checkmark & \checkmark & \checkmark & -- & -- \\
\WhyEnf{}~\cite{Hublet2024}                                               & \checkmark & \checkmark & -- & \checkmark & \checkmark & \checkmark & -- & -- \\
Cedar~\cite{Cutler2024}                                                & \checkmark & -- & -- & \checkmark & -- & -- & \checkmark & \checkmark \\
\EnfGuard{}~\cite{Hublet2025}                                             & \checkmark & \checkmark & \checkmark & \checkmark & \checkmark & \checkmark & -- & -- \\
\bottomrule
\end{tabular}\hfill
\begin{tabular}[t]{lcccccccc}
\toprule
\emph{Agentic} & \rot{First-order\\data} & \rot{Metric\\time} & \rot{Aggregations} & \rot{Suppression}
 & \rot{Causation} & \rot{Proactive\\deadlines} & \rot{Sub-ms\\latency} & \rot{Mechanized\\proofs} \\
\midrule
AgentSpec~\cite{Wang2026}        & $\sim$ & -- & -- & \checkmark & $\sim$ & -- & \checkmark & -- \\
Agent-C~\cite{Kamath2025}        & \checkmark & -- & -- & \checkmark & -- & -- & -- & -- \\
Dogwood~\cite{dogwood,agentcore} & \checkmark & \checkmark & $\sim$ & \checkmark & -- & -- & $\sim$ & -- \\
FORGE~\cite{Palumbo2026}         & \checkmark & -- & -- & \checkmark & -- & -- & -- & -- \\
C-Trace~\cite{Kahani2026}        & \checkmark & -- & -- & \checkmark & -- & -- & -- & -- \\
VIGIL~\cite{Li2026}              & \checkmark & -- & -- & \checkmark & -- & -- & -- & -- \\
ContrAgent~\cite{Xiao2026}       & -- & -- & $\sim$ & \checkmark & -- & -- & \checkmark & -- \\
\midrule
\EnfFlash{} & \checkmark & \checkmark & \checkmark & \checkmark & \checkmark & \checkmark & \checkmark & \checkmark \\
\bottomrule
\end{tabular}
\end{table}

Runtime enforcement was introduced by Schneider in his seminal work on security
automata~\cite{Schneider2000}. Enforcement initially focused on execution cutting, i.e., terminating a system to prevent violations. Suppression~\cite{Bauer2002,FredriksonJJRPSY12} and edit automata~\cite{Ligatti2005} were later introduced to
allow for suppression and causation of events, respectively. However, these enforcers were still allowed to (unrealistically) buffer events before making decisions. This capability was
later dropped~\cite{NgoMMP15}, and other capabilities, such as delaying
events~\cite{Pinisetty2014,Falcone2016} and code inspection~\cite{PinisettyPTJFM17}, were
considered. In parallel, fine-grained enforcer capabilities
were introduced, distinguishing between causable or suppressable and
only-observable events~\cite{Basin2013}. 
Time-aware enforcers rely on timed automata as both a policy language and
mechanism~\cite{Falcone2021,FalconeP19,Pinisetty2015,Renard2017} or as the translation target for logics 
such as MITL~\cite{Maler2006}.
The first proactive enforcement algorithm was developed for DCR graphs~\cite{Basin2024},
a propositional formalism~\cite{Hildebrandt2013}.

\emph{First-order temporal} enforcers exist for LTL-FO$^+$,
a first-order variant of future-only LTL~\cite{HalleV12},
first-order Hennessy-Milner Logic with recursion (non-metric, suppression-only)~\cite{Aceto2023,Aceto2023b},
as well as for several fragments of MFOTL, both past-only~\cite{Hublet2022} and with past and
future~\cite{Hublet2024,Hublet2025}.

Runtime enforcement as presented here is generic and thus provides theoretical underpinning 
for tools in many specialized domains, like access and usage control~\cite{KrsticNB24,Cutler2024}, 
Android permission management~\cite{SadeghiJGBM18}, IoT~\cite{CelikTM19}, and LLM agent security~\cite{Wang2026},
modulo appropriate instrumentation~\cite{Hublet2025b}.

Recent years have seen intense research activity around runtime guardrails for LLM agents.
Most enforcers in this space inspect each call in isolation:
AgentSpec~\cite{Wang2026} evaluates trigger--predicate--action rules;
Progent~\cite{Shi2025} and Conseca~\cite{Tsai2025}
restrict tool names and arguments to those needed by the user's task; and
CaMeL~\cite{Debenedetti2026} and FIDES~\cite{Costa2025} track the provenance of
values through capabilities and information-flow labels. LLM-based
guards~\cite{Xiang2024,Chen2025} give no deterministic guarantees.
Few systems support temporal policies:
Agent-C~\cite{Kamath2025} and VIGIL~\cite{Li2026} check
ordering constraints on the trace with an SMT solver, FORGE~\cite{Palumbo2026}
evaluates Datalog policies, ContrAgent~\cite{Xiao2026} enforces assume-guarantee
contracts in propositional LTL$_f$ with counters, and C-Trace~\cite{Kahani2026}
checks GDPR requirements expressed in MFOTL on agent traces.
Amazon's Bedrock AgentCore supports history-dependent policies written
in Dogwood~\cite{agentcore,dogwood}. None of these systems enforces real-time
constraints over data-carrying events by both suppressing and causing events, e.g.,
limiting the amount sent to each payee within 24 hours while notifying the
account owner of every sensitive change, as in our banking case study
(Section~\ref{sec:eval-app3}). Table~\ref{tab:rw} summarizes these differences.

\section{Conclusion}

We presented \EnfFlash{}, a novel first-order enforcer that is one to two orders of
magnitude faster than the previous state-of-the-art, with comparable expressivity.
\EnfFlash{} is also up to an order of magnitude faster than the recently proposed
Dogwood tool, developed for agentic security tasks; moreover, it is substantially more expressive and scales better with large logs.
Even when deployed with very large policies, such as  a 400-line MFOTL formula in a web application, \EnfFlash{} induces
only 1--15\,ms of overhead. These results support the thesis that  first-order runtime enforcement
tools have become a mature alternative to ad-hoc frameworks for software security,
including in the rapidly evolving field of agent security.


\bibliographystyle{ACM-Reference-Format}
\bibliography{enfflash}

@article{Schneider2000,
  title={Enforceable security policies},
  author={Schneider, Fred B.},
  journal={{ACM} Trans. Inf. Syst. Secur. {(TISSEC)}},
  volume={3},
  number={1},
  pages={30--50},
  year={2000},
  publisher={ACM New York, NY, USA}
}

@inproceedings{Bauer2002,
  title={More enforceable security policies},
  author={Bauer, Lujo and Ligatti, Jarred and Walker, David},
  booktitle={Workshop on Foundations of Computer Security ({FCS})},
  year={2002}
}

@article{Chomicki1995,
  title={Efficient checking of temporal integrity constraints using bounded history encoding},
  author={Chomicki, Jan},
  journal={ACM Transactions on Database Systems (TODS)},
  volume={20},
  number={2},
  pages={149--186},
  year={1995},
  publisher={ACM New York, NY, USA}
}

@inproceedings{Hublet2025,
  title={Scaling up proactive enforcement},
  author={Hublet, Fran{\c{c}}ois and Lima, Leonardo and Basin, David and Krsti{\'c}, Sr{\dj}an and Traytel, Dmitriy},
    editor       = {Ruzica Piskac and
                  Zvonimir Rakamaric},
  booktitle={37th International Conference on Computer Aided Verification ({CAV})},
  series = {LNCS},
    volume       = {15933},
  pages={370--392},
  year={2025},
  publisher={Springer}
}

@inproceedings{Hublet2024,
  title={Proactive real-time first-order enforcement},
  author={Hublet, Fran{\c{c}}ois and Lima, Leonardo and Basin, David and Krsti{\'c}, Sr{\dj}an and Traytel, Dmitriy},
    editor       = {Arie Gurfinkel and
                  Vijay Ganesh},
  booktitle={36th International Conference on Computer Aided Verification ({CAV})},
    series = {LNCS},
      volume       = {14682},
  pages={156--181},
  year={2024},
  publisher={Springer}
}

@article{Cutler2024,
  title={Cedar: A new language for expressive, fast, safe, and analyzable authorization},
  author={Joseph W. Cutler and
                  Craig Disselkoen and
                  Aaron Eline and
                  Shaobo He and
                  Kyle Headley and
                  Michael Hicks and
                  Kesha Hietala and
                  Eleftherios Ioannidis and
                  John H. Kastner and
                  Anwar Mamat and
                  Darin McAdams and
                  Matt McCutchen and
                  Neha Rungta and
                  Emina Torlak and
                  Andrew Wells},
  journal={Proc. {ACM} Program. Lang.},
  volume={8},
  number={OOPSLA1},
  pages={670--697},
  year={2024},
  publisher={ACM New York, NY, USA}
}

@inproceedings{Hublet2022,
  title={Real-time policy enforcement with metric first-order temporal logic},
  author={Hublet, Fran{\c{c}}ois and Basin, David and Krsti{\'c}, Sr{\dj}an},
    editor       = {Vijayalakshmi Atluri and
                  Roberto Di Pietro and
                  Christian Damsgaard Jensen and
                  Weizhi Meng},
  booktitle={27th European Symposium on Research in Computer Security ({ESORICS})},
  pages={211--232},
  year={2022},
      series       = {LNCS},
  publisher={Springer}
}

@inproceedings{Havelund2024,
  title={Operational and declarative runtime verification (keynote)},
  author={Havelund, Klaus and Omer, Moran and Peled, Doron},
    editor       = {Davide Ancona and
                  Giorgio Audrito},
    publisher    = {{ACM}},
  booktitle={7th ACM International Workshop on Verification and Monitoring at Runtime Execution ({VORTEX})},
  pages={3--12},
  year={2024}
}

@inproceedings{Nguyen2026,
  title={{Models Trump Code}: An Empirical Analysis of {GDPR} Compliance},
  author={Nguy\~{\^e}n, Hoàng and Krstić, Srđan and Basin, David},
  booktitle={2026 ACM SIGSAC Conference on Computer and Communications Security},
  year={2026},
}

@inproceedings{BasinGKNO23,
  author       = {David Basin and
                  Juan Guarnizo and
                  Sr\dj{}an Krsti\'c and
                  Hoàng Nguyen and
                  Mart{\'{\i}}n Ochoa},
  editor       = {Weizhi Meng and
                  Christian Damsgaard Jensen and
                  Cas Cremers and
                  Engin Kirda},
  title        = {Is Modeling Access Control Worth It?},
  booktitle    = {{ACM} {SIGSAC} Conference on Computer and
                  Communications Security ({CCS})},
  pages        = {2830--2844},
  publisher    = {{ACM}},
  year         = {2023},
  xurl          = {https://doi.org/10.1145/3576915.3623196},
  doi          = {10.1145/3576915.3623196},
}

@inproceedings{Hublet2025b,
  title={Instrumenting runtime enforcement},
  author={Hublet, Fran{\c{c}}ois and Basin, David and Hu, Linda and Krsti{\'c}, Sr{\dj}an and Reese, Lennard},
  booktitle={25th International Conference on Runtime Verification ({RV})},
  pages={160--180},
  year={2025},
  series={LNCS},
  publisher={Springer},
  volume       = {16087},
  doi          = {10.1007/978-3-032-05435-7\_10},
  xurl          = {https://doi.org/10.1007/978-3-032-05435-7\_10},
}

@inproceedings{Basin2011,
  title={{MONPOLY}: Monitoring usage-control policies},
  author={Basin, David and Harvan, Mat{\'u}{\v{s}} and Klaedtke, Felix and Z{\u{a}}linescu, Eugen},
  booktitle={2nd International Conference on Runtime Verification ({RV})},
  pages={360--364},
  year={2011},
  series={LNCS},
  volume       = {7186},
  publisher={Springer},
}

@article{Basin2024,
  title={Proactive enforcement of provisions and obligations},
  author={Basin, David and Debois, S{\o}ren and Hildebrandt, Thomas},
  journal={Journal of Computer Security},
  volume={32},
  number={3},
  pages={247--289},
  year={2024},
  publisher={SAGE Publications Sage UK: London, England}
}

@inproceedings{Pnueli1977,
  title={The temporal logic of programs},
  author={Pnueli, Amir},
  booktitle={18th Annual Symposium on Foundations of Computer Science ({FOCS})},
  pages={46--57},
  year={1977},
  organization={IEEE}
}

@article{Koymans1990,
  title={Specifying real-time properties with metric temporal logic},
  author={Koymans, Ron},
  journal={Real-time systems},
  volume={2},
  number={4},
  pages={255--299},
  year={1990},
  publisher={Springer}
}

@book{Nielsen1993,
  title={Usability engineering},
  author={Nielsen, Jakob},
  year={1993},
  publisher={AP Professional}
}

@inproceedings{FredriksonJJRPSY12,
  author       = {Matthew Fredrikson and
                  Richard Joiner and
                  Somesh Jha and
                  Thomas W. Reps and
                  Phillip A. Porras and
                  Hassen Sa{\"{\i}}di and
                  Vinod Yegneswaran},
  editor       = {P. Madhusudan and
                  Sanjit A. Seshia},
  title        = {Efficient Runtime Policy Enforcement Using Counterexample-Guided Abstraction
                  Refinement},
  booktitle    = {24th International Conference on Computer Aided Verification ({CAV})},
  series       = {LNCS},
  volume       = {7358},
  pages        = {548--563},
  publisher    = {Springer},
  year         = {2012},
}

@article{Ligatti2005,
  title={Edit automata: Enforcement mechanisms for run-time security policies},
  author={Ligatti, Jay and Bauer, Lujo and Walker, David},
  journal={International Journal of Information Security},
  volume={4},
  pages={2--16},
  year={2005},
  publisher={Springer}
}

@inproceedings{NgoMMP15,
  author       = {Minh Ngo and
                  Fabio Massacci and
                  Dimiter Milushev and
                  Frank Piessens},
  editor       = {Sriram K. Rajamani and
                  David Walker},
  title        = {Runtime Enforcement of Security Policies on Black Box Reactive Programs},
  booktitle    = {42nd {ACM} {SIGPLAN-SIGACT} Symposium on
                  Principles of Programming Languages {(POPL)}},
  pages        = {43--54},
  publisher    = {{ACM}},
  year         = {2015},
}

@article{Pinisetty2014,
  title={Runtime enforcement of timed properties revisited},
  author={Pinisetty, Srinivas and Falcone, Ylies and J{\'e}ron, Thierry and Marchand, Herv{\'e} and Rollet, Antoine and Nguena Timo, Omer},
  journal={Formal Methods Syst. Des.},
  volume={45},
  pages={381--422},
  year={2014},
  publisher={Springer}
}

@article{PinisettyPTJFM17,
  author       = {Srinivas Pinisetty and
                  Viorel Preoteasa and
                  Stavros Tripakis and
                  Thierry J{\'{e}}ron and
                  Yli{\`{e}}s Falcone and
                  Herv{\'{e}} Marchand},
  title        = {Predictive runtime enforcement},
  journal      = {Formal Methods Syst. Des.},
  volume       = {51},
  number       = {1},
  pages        = {154--199},
  year         = {2017},
}

@article{Falcone2016,
  title={Runtime enforcement of regular timed properties by suppressing and delaying events},
  author={Falcone, Yli{\`e}s and J{\'e}ron, Thierry and Marchand, Herv{\'e} and Pinisetty, Srinivas},
  journal={Science of Computer Programming},
  volume={123},
  pages={2--41},
  year={2016},
  publisher={Elsevier}
}

@article{basin2013,
  title={Enforceable security policies revisited},
  author={Basin, David and Jug{\'e}, Vincent and Klaedtke, Felix and Z{\u{a}}linescu, Eugen},
  journal={{ACM} Trans. Inf. Syst. Secur. ({TISSEC})},
  volume={16},
  number={1},
  pages={1--26},
  year={2013},
  publisher={ACM New York, NY, USA}
}

@article{Falcone2021,
  author    = {Yli{\`{e}}s Falcone and
               Sr\dj{}an Krsti\'c and
               Giles Reger and
               Dmitriy Traytel},
  title     = {A taxonomy for classifying runtime verification tools},
  journal   = {Int. J. Softw. Tools Technol. Transf.},
  volume    = {23},
  number    = {2},
  pages     = {255--284},
  year      = {2021},
  xurl       = {https://doi.org/10.1007/s10009-021-00609-z},
  xdoi       = {10.1007/s10009-021-00609-z},
}

@inproceedings{FalconeP19,
  author       = {Yli{\`{e}}s Falcone and
                  Srinivas Pinisetty},
  editor       = {Bernd Finkbeiner and
                  Leonardo Mariani},
  title        = {On the Runtime Enforcement of Timed Properties},
  booktitle    = {19th International Conference on Runtime Verification, ({RV})},
  series       = {LNCS},
  volume       = {11757},
  pages        = {48--69},
  publisher    = {Springer},
  year         = {2019},
  xurl          = {https://doi.org/10.1007/978-3-030-32079-9\_4},
  xdoi          = {10.1007/978-3-030-32079-9\_4},
}

@inproceedings{Pinisetty2015,
  title={Ti{PEX}: A tool chain for timed property enforcement during execution},
  author={Pinisetty, Srinivas and Falcone, Yli{\`e}s and J{\'e}ron, Thierry and Marchand, Herv{\'e}},
  booktitle={International Conference on Runtime Verification (RV)},
  pages={306--320},
  year={2015},
  publisher={Springer}
}

@inproceedings{Renard2017,
  author    = {Matthieu Renard and
               Antoine Rollet and
               Yli{\`{e}}s Falcone},
  editor    = {Nina Yevtushenko and
               {Ana Rosa} Cavalli and
               H{\"{u}}sn{\"{u}} Yenig{\"{u}}n},
  title     = {{GREP:} Games for the Runtime Enforcement of Properties},
  booktitle = {International
               Conference on Testing Software and Systems (ICTSS)},
  series    = {LNCS},
  volume    = {10533},
  pages     = {259--275},
  publisher = {Springer},
  year      = {2017},
  xurl       = {https://doi.org/10.1007/978-3-319-67549-7\_16},
  xdoi       = {10.1007/978-3-319-67549-7\_16},
}

@inproceedings{Maler2006,
  title={From {MITL} to timed automata},
  author={Maler, Oded and Ni\v{c}kovi\'c, Dejan and Pnueli, Amir},
  editor       = {Eugene Asarin and
                  Patricia Bouyer},
  booktitle={International conference on formal modeling and analysis of timed systems ({FORMATS})},
  series       = {LNCS},
  volume       = {4202},
  pages={274--289},
  year={2006},
  publisher={Springer}
}

@article{HalleV12,
  author       = {Sylvain Hall{\'{e}} and
                  Roger Villemaire},
  title        = {Runtime Enforcement of Web Service Message Contracts with Data},
  journal      = {{IEEE} Trans. Serv. Comput.},
  volume       = {5},
  number       = {2},
  pages        = {192--206},
  year         = {2012},
}

@article{Aceto2023b,
  title={On first-order runtime enforcement of branching-time properties},
  author={Aceto, Luca and Cassar, Ian and Francalanza, Adrian and Ing{\'o}lfsd{\'o}ttir, Anna},
  journal={Acta Informatica},
  pages={1--67},
  year={2023},
  publisher={Springer}
}

@article{Aceto2023,
  title={Bidirectional Runtime Enforcement of First-Order Branching-Time Properties},
  author={Aceto, Luca and Cassar, Ian and Francalanza, Adrian and Ing{\'o}lfsd{\'o}ttir, Anna},
  journal={Logical Methods in Computer Science},
  volume={19},
  year={2023},
  publisher={Episciences. org}
}

@article{Hildebrandt2013,
  title={Contracts for cross-organizational workflows as timed dynamic condition response graphs},
  author={Hildebrandt, Thomas and Mukkamala, Raghava Rao and Slaats, Tijs and Zanitti, Francesco},
  journal={The Journal of Logic and Algebraic Programming},
  volume={82},
  number={5-7},
  pages={164--185},
  year={2013},
  publisher={Elsevier}
}

@article{KrsticNB24,
  author       = {Sr\dj{}an Krsti\'c and
                  Hoang Nguyen and
                  David Basin},
  title        = {Model-driven Privacy},
  journal      = {Proc. Priv. Enhancing Technol.},
  volume       = {2024},
  number       = {1},
  pages        = {314--329},
  year         = {2024},
  xurl          = {https://doi.org/10.56553/popets-2024-0018},
  doi          = {10.56553/POPETS-2024-0018},
}

@inproceedings{CelikTM19,
  author       = {Z. Berkay Celik and
                  Gang Tan and
                  Patrick D. McDaniel},
  title        = {{IoTGuard}: Dynamic Enforcement of Security and Safety Policy in Commodity
                  IoT},
  booktitle    = {26th Annual Network and Distributed System Security Symposium ({NDSS})},
  publisher    = {The Internet Society},
  year         = {2019},
  xurl          = {https://www.ndss-symposium.org/ndss-paper/iotguard-dynamic-enforcement-of-security-and-safety-policy-in-commodity-iot/},
}

@inproceedings{Wang2026,
  title={{AgentSpec}: Customizable runtime enforcement for safe and reliable LLM agents},
  author={Wang, Haoyu and Poskitt, Christopher M. and Sun, Jun},
  booktitle={Proceedings of the 2026 IEEE/ACM 48th International Conference on Software Engineering},
  pages={2938--2950},
  year={2026}
}

@inproceedings{SadeghiJGBM18,
  author       = {Alireza Sadeghi and
                  Reyhaneh Jabbarvand and
                  Negar Ghorbani and
                  Hamid Bagheri and
                  Sam Malek},
  editor       = {Michel Chaudron and
                  Ivica Crnkovic and
                  Marsha Chechik and
                  Mark Harman},
  title        = {A temporal permission analysis and enforcement framework for {Android}},
  booktitle    = {40th International Conference on Software Engineering,
                  ({ICSE})},
  pages        = {846--857},
  publisher    = {{ACM}},
  year         = {2018},
  xurl          = {https://doi.org/10.1145/3180155.3180172},
  doi          = {10.1145/3180155.3180172},
}

@inproceedings{Krstic2020,
  title={A benchmark generator for online first-order monitoring},
  author={Krsti{\'c}, Sr{\dj}an and Schneider, Joshua},
  booktitle={20th International Conference on Runtime Verification ({RV})},
  pages={482--494},
  year={2020},
  organization={Springer}
}

@inproceedings{Moura2021,
  title={The lean 4 theorem prover and programming language},
  author={Moura, Leonardo de and Ullrich, Sebastian},
  booktitle={International Conference on Automated Deduction},
  pages={625--635},
  year={2021},
  organization={Springer}
}

@misc{dogwood,
  title = {Introducing Dogwood: runtime verification for AI agents},
  author = {Marc Brooker and Joseph Tassarotti and Jean-Baptiste Tristan},
  url = {https://aws.amazon.com/blogs/opensource/introducing-dogwood-runtime-verification-for-ai-agents/}
}

@inproceedings{HiltyBP05,
  author       = {Manuel Hilty and
                  David Basin and
                  Alexander Pretschner},
  editor       = {Sabrina De Capitani di Vimercati and
                  Paul F. Syverson and
                  Dieter Gollmann},
  title        = {On Obligations},
  booktitle    = {10th European Symposium on Research
                  in Computer Security ({ESORICS})},
  series       = {LNCS},
  volume       = {3679},
  pages        = {98--117},
  publisher    = {Springer},
  year         = {2005},
  url          = {https://doi.org/10.1007/11555827\_7},
  doi          = {10.1007/11555827\_7},
}

@inproceedings{Debenedetti2024,
  author    = {Edoardo Debenedetti and Jie Zhang and Mislav Balunovi{\'c} and
               Luca Beurer-Kellner and Marc Fischer and Florian Tram{\`e}r},
  title     = {{AgentDojo}: A Dynamic Environment to Evaluate Prompt Injection Attacks
               and Defenses for {LLM} Agents},
  booktitle = {Advances in Neural Information Processing Systems 37 ({NeurIPS}),
               Datasets and Benchmarks Track},
  year      = {2024}
}

@misc{Xiao2026,
  author        = {Yifeng Xiao and Pierluigi Nuzzo},
  title         = {Symbolic Temporal Supervision of {LLM} Agents Using Contracts},
  year          = {2026},
  eprint        = {2609.18128},
  archivePrefix = {arXiv}
}

@misc{Kamath2025,
  author        = {Adharsh Kamath and Sishen Zhang and Calvin Xu and Shubham Ugare and
                   Gagandeep Singh and Sasa Misailovic},
  title         = {Enforcing Temporal Constraints for {LLM} Agents},
  year          = {2025},
  eprint        = {2512.23738},
  archivePrefix = {arXiv}
}

@misc{Shi2025,
  author        = {Tianneng Shi and Jingxuan He and Zhun Wang and Hongwei Li and Linyu Wu and
                   Wenbo Guo and Dawn Song},
  title         = {Progent: Programmable Privilege Control for {LLM} Agents},
  year          = {2025},
  eprint        = {2504.11703},
  archivePrefix = {arXiv}
}

@inproceedings{Debenedetti2026,
  title={Defeating prompt injections by design},
  author={Debenedetti, Edoardo and Shumailov, Ilia and Fan, Tianqi and Hayes, Jamie and Carlini, Nicholas and Fabian, Daniel and Kern, Christoph and Shi, Chongyang and Terzis, Andreas and Tram{\`e}r, Florian},
  booktitle={2026 IEEE Conference on Secure and Trustworthy Machine Learning (SaTML)},
  pages={587--618},
  year={2026},
  organization={IEEE}
}

@misc{Costa2025,
  author        = {Manuel Costa and Boris K{\"o}pf and Aashish Kolluri and Andrew Paverd and
                   Mark Russinovich and Ahmed Salem and Shruti Tople and Lukas Wutschitz and
                   Santiago Zanella-B{\'e}guelin},
  title         = {Securing {AI} Agents with Information-Flow Control},
  year          = {2025},
  eprint        = {2505.23643},
  archivePrefix = {arXiv}
}

@inproceedings{Tsai2025,
  author    = {Lillian Tsai and Eugene Bagdasarian},
  title     = {Contextual Agent Security: A Policy for Every Purpose},
  booktitle = {Workshop on Hot Topics in Operating Systems ({HotOS})},
  year      = {2025}
}

@misc{Xiang2024,
  title={Guardagent: Safeguard llm agents by a guard agent via knowledge-enabled reasoning},
  author={Xiang, Zhen and Zheng, Linzhi and Li, Yanjie and Hong, Junyuan and Li, Qinbin and Xie, Han and Zhang, Jiawei and Xiong, Zidi and Xie, Chulin and Yang, Carl and others},
  eprint={2406.09187},
  archivePrefix={arXiv},
  year={2024}
}

@misc{Chen2025,
  title={Shieldagent: Shielding agents via verifiable safety policy reasoning},
  author={Chen, Zhaorun and Kang, Mintong and Li, Bo},
  eprint={2503.22738},
  archivePrefix={arXiv},
  year={2025}
}

@misc{Li2026,
  author        = {Ying Li and Yanju Chen and Hongbo Wen and Bosi Zhang and Hanzhi Liu and
                   Peiran Wang and Yu Feng and Yuan Tian},
  title         = {{VIGIL}: Runtime Enforcement of Behavioral Specifications in {AI} Agent Skills},
  year          = {2026},
  eprint        = {2606.26524},
  archivePrefix = {arXiv}
}

@misc{Kahani2026,
  author        = {Nafiseh Kahani and Masoud Barati and Diana Addae},
  title         = {Runtime Compliance Verification for {AI} Agents},
  year          = {2026},
  eprint        = {2606.19242},
  archivePrefix = {arXiv}
}

@misc{Palumbo2026,
  author        = {Nils Palumbo and Sarthak Choudhary and Jihye Choi and Guy Amir and
                   Prasad Chalasani and Somesh Jha},
  title         = {Formal Policy Enforcement for Real-World Agentic Systems},
  year          = {2026},
  eprint        = {2602.16708},
  archivePrefix = {arXiv}
}

@misc{agentcore,
  author = {{Amazon Web Services}},
  title  = {Temporal policies -- {Amazon} {Bedrock} {AgentCore}},
  year   = {2026},
  url    = {https://docs.aws.amazon.com/bedrock-agentcore/latest/devguide/policy-temporal.html}
}

@misc{extended,
  title  = {Practical Runtime Enforcement of First-Order Temporal Requirements (extended version)},
  author = {Anonymous},
  year   = {2026},
  url    = {https://anonymous.4open.science/r/EnfFlash-6400/docs/enfflash-extended.pdf}
}

\ifextended
\appendix

\section{Detailed evaluation results}
\label{app:results}

Tables~\ref{tab:micro}--\ref{tab:app3} give the results summarized in
Figure~\ref{fig:eval}.

\begin{table}[t]
\setlength{\tabcolsep}{2pt}
\centering
\caption{Latency in ms---$\mathit{mean}\ (\mathit{max})$---on the largest logs on the \EnfGuard{}
  benchmark suite at maximal real-time acceleration}
\label{tab:micro}
\small
\begin{tabular}{lrrrrrrrr@{\hspace{6pt}\vrule\hspace{1.5pt}\vrule\hspace{6pt}}rr}
\toprule
 & \multicolumn{8}{c@{\hspace{6pt}\vrule\hspace{1.5pt}\vrule\hspace{6pt}}}{Enforcement} & \multicolumn{2}{c}{Monitoring} \\
Policy & \multicolumn{2}{c}{EnfFlash} & \multicolumn{2}{c}{Enfpoly} & \multicolumn{2}{c}{EnfGuard} & \multicolumn{2}{c@{\hspace{6pt}\vrule\hspace{1.5pt}\vrule\hspace{6pt}}}{Dogwood} & \multicolumn{2}{c}{Monpoly} \\
\midrule
\multicolumn{9}{l@{\hspace{6pt}\vrule\hspace{1.5pt}\vrule\hspace{6pt}}}{\emph{\textsc{gdpr} (timeout = 60 s)}} & \multicolumn{2}{l}{} \\
consent & \textbf{0.08} & (1) & 0.11 & (1) & 1.32 & (6) & 0.19 & (1) & 0.14 & (2) \\
deletion & \textbf{0.08} & (1) & -- &  & 0.32 & (3) & -- &  & 0.09 & (1) \\
gdpr & \textbf{0.09} & (2) & -- &  & 3.97 & (13) & -- &  & 0.11 & (3) \\
information & \textbf{0.08} & (1) & -- &  & 1.20 & (7) & -- &  & 0.09 & (2) \\
lawfulness & \textbf{0.08} & (1) & 0.09 & (2) & 1.46 & (7) & 0.19 & (2) & 0.10 & (1) \\
limitation & \textbf{0.08} & (1) & -- &  & 0.35 & (59) & -- &  & -- &  \\
sharing & \textbf{0.08} & (1) & -- &  & 1.14 & (5) & -- &  & 0.10 & (1) \\
\midrule
\multicolumn{9}{l@{\hspace{6pt}\vrule\hspace{1.5pt}\vrule\hspace{6pt}}}{\emph{\textsc{fun} (timeout = 60 s)}} & \multicolumn{2}{l}{} \\
consent & \textbf{0.08} & (2) & -- &  & 0.83 & (5) & -- &  & -- &  \\
consent\_management & \textbf{0.08} & (1) & -- &  & 0.29 & (2) & -- &  & -- &  \\
deletion & \textbf{0.08} & (1) & -- &  & 0.32 & (3) & -- &  & -- &  \\
gdpr & \textbf{0.09} & (1) & -- &  & 3.36 & (12) & -- &  & -- &  \\
information & \textbf{0.08} & (2) & -- &  & 1.55 & (9) & -- &  & -- &  \\
sharing & \textbf{0.08} & (2) & -- &  & 1.23 & (5) & -- &  & -- &  \\
\midrule
\multicolumn{9}{l@{\hspace{6pt}\vrule\hspace{1.5pt}\vrule\hspace{6pt}}}{\emph{\textsc{cluster} (timeout = 60 s)}} & \multicolumn{2}{l}{} \\
dbscan & \textbf{3.78} & (12) & -- &  & t.o. &  & -- &  & -- &  \\
grubbs & \textbf{2.73} & (7) & -- &  & t.o. &  & -- &  & -- &  \\
\midrule
\multicolumn{9}{l@{\hspace{6pt}\vrule\hspace{1.5pt}\vrule\hspace{6pt}}}{\emph{\textsc{agg} (timeout = 120 s)}} & \multicolumn{2}{l}{} \\
p1 & \textbf{0.52} & (4) & -- &  & t.o. &  & 8.71 & (45) & 0.19 & (3) \\
p2 & \textbf{0.53} & (3) & -- &  & t.o. &  & t.o. &  & 0.32 & (3) \\
p3 & \textbf{0.53} & (2) & -- &  & t.o. &  & t.o. &  & 0.39 & (3) \\
p4 & \textbf{1.35} & (15) & -- &  & t.o. &  & -- &  & 0.52 & (3) \\
p5 & \textbf{0.81} & (3) & -- &  & t.o. &  & -- &  & 0.21 & (2) \\
p6 & \textbf{0.36} & (4) & -- &  & t.o. &  & -- &  & 0.44 & (3) \\
\midrule
\multicolumn{9}{l@{\hspace{6pt}\vrule\hspace{1.5pt}\vrule\hspace{6pt}}}{\emph{\textsc{nokia} (timeout = 120 s)}} & \multicolumn{2}{l}{} \\
del-1-2 & \textbf{0.25} & (3) & -- &  & t.o. &  & -- &  & 0.20 & (19) \\
del-2-3 & \textbf{0.25} & (12) & -- &  & t.o. &  & -- &  & 0.12 & (3) \\
del-3-2 & \textbf{0.25} & (32) & -- &  & t.o. &  & -- &  & 0.13 & (3) \\
delete & 0.25 & (22) & \textbf{0.13} & (3) & 0.79 & (11) & 0.23 & (4) & 0.13 & (4) \\
ins-1-2 & \textbf{0.98} & (34) & -- &  & t.o. &  & -- &  & 1.06 & (112) \\
ins-2-3 & \textbf{0.33} & (3) & -- &  & t.o. &  & -- &  & 0.41 & (9) \\
ins-3-2 & \textbf{0.35} & (22) & -- &  & t.o. &  & -- &  & 0.40 & (7) \\
insert & 0.26 & (14) & \textbf{0.20} & (3) & 1.30 & (25) & t.o. &  & 0.20 & (5) \\
script1 & 0.35 & (23) & \textbf{0.15} & (2) & t.o. &  & -- &  & 0.28 & (3) \\
select & 0.25 & (12) & \textbf{0.12} & (3) & 0.82 & (12) & 0.25 & (4) & 0.13 & (3) \\
update & 0.25 & (13) & \textbf{0.13} & (4) & 0.85 & (15) & 0.28 & (5) & 0.13 & (4) \\
\midrule
\multicolumn{9}{l@{\hspace{6pt}\vrule\hspace{1.5pt}\vrule\hspace{6pt}}}{\emph{\textsc{ic} (timeout = 600 s)}} & \multicolumn{2}{l}{} \\
block\_validation\_latency & 0.21 & (32) & -- &  & \textbf{0.15} & (7) & -- &  & t.o. &  \\
clean\_logs & \textbf{0.16} & (23) & -- &  & t.o. &  & -- &  & 0.12 & (5) \\
finalization\_consistency & 0.17 & (32) & -- &  & t.o. &  & \textbf{0.15} & (7) & 0.08 & (4) \\
finalized\_height & t.o. &  & -- &  & \textbf{0.17} & (5) & -- &  & t.o. &  \\
logging\_behavior & 0.78 & (98) & -- &  & \textbf{0.11} & (25) & -- &  & t.o. &  \\
reboot\_count & \textbf{0.10} & (23) & -- &  & t.o. &  & -- &  & 0.11 & (4) \\
replica\_divergence & \textbf{0.20} & (32) & -- &  & 0.23 & (5) & -- &  & 0.08 & (4) \\
unauthorized\_connections & 0.36 & (95) & -- &  & t.o. &  & \textbf{0.10} & (6) & t.o. &  \\
\bottomrule
\end{tabular}

\end{table}

\begin{table}[t]
\setlength{\tabcolsep}{2pt}
\centering
\caption{GDPRSocial page latency in ms (+overhead to non-enforced baseline) for representative application views}
\label{tab:app1}
\small
\begin{tabular}{lrrrrrrrr}
\toprule
View & \multicolumn{2}{c}{$n=10^{2}$} & \multicolumn{2}{c}{$n=10^{3}$} & \multicolumn{2}{c}{$n=10^{4}$} & \multicolumn{2}{c}{$n=10^{5}$} \\
\midrule
View timeline & 19.98 & +9 & 24.63 & +10 & 25.69 & +11 & 26.32 & +15 \\
Post message & 6.82 & +4 & 8.84 & +5 & 15.61 & +12 & 13.31 & +8 \\
Grant consent & 7.00 & +4 & 7.80 & +4 & 10.44 & +6 & 9.68 & +6 \\
Revoke consent & 6.14 & +3 & 6.87 & +4 & 10.63 & +7 & 9.14 & +6 \\
View privacy notices & 7.56 & +2 & 6.68 & +1 & 6.27 & +2 & 7.84 & +2 \\
Request personal data & 5.83 & +1 & 8.79 & +5 & 6.40 & +2 & 7.19 & +3 \\
\bottomrule
\end{tabular}

\end{table}

\begin{table}[t]
\setlength{\tabcolsep}{2.4pt}
\centering
\caption{\textsc{EventManager} page latency in ms: \EnfFlash{} (E) vs. Dogwood (D) vs. Cedar (C). t.o.: timeout (300\,s per request)}
\label{tab:app2}
\footnotesize
\begin{tabular}{l|rrr|rrr|rrr|rrr|rrr}
\toprule
View & \multicolumn{3}{c|}{$n=1$} & \multicolumn{3}{c|}{$n=10$} & \multicolumn{3}{c|}{$n=100$} & \multicolumn{3}{c|}{$n=1000$} & \multicolumn{3}{c}{$n=10000$} \\
 & \multicolumn{1}{c}{E} & \multicolumn{1}{c}{D} & \multicolumn{1}{c|}{C} & \multicolumn{1}{c}{E} & \multicolumn{1}{c}{D} & \multicolumn{1}{c|}{C} & \multicolumn{1}{c}{E} & \multicolumn{1}{c}{D} & \multicolumn{1}{c|}{C} & \multicolumn{1}{c}{E} & \multicolumn{1}{c}{D} & \multicolumn{1}{c|}{C} & \multicolumn{1}{c}{E} & \multicolumn{1}{c}{D} & \multicolumn{1}{c}{C} \\
\midrule
View events & \textbf{27.79} & 29.44 & 1 s & \textbf{1 s} & 1 s & 11 s & \textbf{13 s} & 15 s & 112 s & \textbf{183 s} & 294 s & t.o. & t.o. & t.o. & t.o. \\
View profile & 33.23 & \textbf{33.15} & 4 s & 36.26 & \textbf{34.60} & 3 s & 43.29 & \textbf{38.06} & 3 s & \textbf{68.04} & 78.70 & t.o. & \textbf{114.60} & 340.63 & t.o. \\
Create event & 17.97 & \textbf{15.83} & 32.48 & 19.23 & \textbf{17.32} & 24.73 & \textbf{20.26} & 32.65 & 97.20 & \textbf{24.19} & 85.68 & t.o. & 32.43 & \textbf{28.41} & t.o. \\
Grant consent & 18.17 & 16.15 & \textbf{14.72} & 19.12 & 21.23 & \textbf{13.04} & \textbf{19.55} & 25.50 & 41.53 & 22.38 & \textbf{18.40} & t.o. & \textbf{29.67} & 30.44 & t.o. \\
Revoke consent & 6.14 & 6.82 & \textbf{5.35} & 9.30 & 6.52 & \textbf{4.96} & 6.46 & 9.17 & \textbf{4.79} & \textbf{7.70} & 8.89 & t.o. & \textbf{17.23} & 17.98 & t.o. \\
Analyze event & \textbf{10.03} & 10.08 & 119.38 & 21.44 & \textbf{19.94} & 112.52 & 23.04 & \textbf{21.07} & 113.15 & \textbf{28.64} & 40.55 & t.o. & \textbf{68.29} & 149.32 & t.o. \\
\bottomrule
\end{tabular}

\end{table}

\begin{table}[t]
\setlength{\tabcolsep}{2pt}
\centering
\caption{LLM banking agent (AgentDojo): attacks prevented among those successful
  without enforcement (Prev.), benign runs with a suppressed call (FP), attacked
  runs whose user task is completed (Util.), and latency in ms---$\mathit{mean}\
  (\mathit{max})$---on the recorded logs. $^\dagger$:~suppression part only}
\label{tab:app3}
\small
\begin{tabular}{ll@{\hspace{6pt}\vrule\hspace{6pt}}rrr@{\hspace{6pt}\vrule\hspace{6pt}}rrrrrr}
\toprule
& Policy & \multicolumn{1}{c}{Prev.} & \multicolumn{1}{c}{FP} & \multicolumn{1}{c@{\hspace{6pt}\vrule\hspace{6pt}}}{Util.} & \multicolumn{2}{c}{EnfFlash} & \multicolumn{2}{c}{EnfGuard} & \multicolumn{2}{c}{Dogwood} \\
\midrule
-- & No enforcement & 0/143 & 0/16 & 126/144 & & & & & & \\
B1 & Payee provenance & 77/143 & 0/16 & 131/144 & \textbf{0.05} & (0.3) & t.o. &  & 0.40 & (5) \\
B2 & No leak in subject & 48/143 & 0/16 & 134/144 & \textbf{0.04} & (0.5) & 222.79 & (1278) & 0.65 & (20) \\
B3 & Daily cap per payee & 32/143 & 0/16 & 128/144 & \textbf{0.04} & (0.2) & 18.93 & (100) & 0.17 & (0.8) \\
B4 & Velocity & 16/143 & 0/16 & 121/144 & \textbf{0.03} & (0.1) & 2.23 & (9) & 0.16 & (0.7) \\
B5 & Standing-order redirect & 15/143 & 0/16 & 128/144 & \textbf{0.04} & (0.2) & 10.42 & (50) & 0.08$^\dagger$ & (1) \\
B6 & Password change & 15/143 & 0/16 & 128/144 & \textbf{0.04} & (0.4) & 10.22 & (52) & 0.07$^\dagger$ & (2) \\
B7 & Cool-down after credential change & 7/143 & 0/16 & 124/144 & \textbf{0.04} & (0.1) & 47.66 & (345) & 0.17 & (2) \\
B8 & Quarantine & 0/143 & 0/16 & 126/144 & \textbf{0.04} & (0.2) & 1.35 & (5) & 0.16$^\dagger$ & (0.6) \\
B9 & Audit & 0/143 & 0/16 & 126/144 & \textbf{0.04} & (0.1) & 0.82 & (4) & -- &  \\
B10 & Invoice consistency & 101/143 & 0/16 & 131/144 & \textbf{0.05} & (0.2) & 77.04 & (546) & 0.37 & (4) \\
\midrule
-- & All (B1--B10) & 142/143 & 0/16 & 55/144 & \textbf{0.08} & (0.3) & t.o. &  & 0.98$^\dagger$ & (28) \\
\bottomrule
\end{tabular}

\end{table}

\section{A type system for the enforceable fragment}
\label{app:types}

\newcommand{\GRD}{\mathsf{GRD}}
\newcommand{\Enum}{\mathbb{G}}
\newcommand{\triggers}[1]{\mathrm{trig}(#1)}

The fragment enforced by \EnfFlash{} is defined implicitly by the compilation of
Section~\ref{sec:mfotl-to-ef}: a formula is supported if let-normalization,
\textsc{TypeLet}, and the enforcement rewriting succeed. We now characterize
this fragment by a type system over formulae in let-normal form, which we call
EF-MFOTL, and show that typability coincides with successful compilation.

\paragraph{Contexts.}
A context $\Gamma$ assigns to each let-bound event $p$ a set of capabilities
$\Gamma(p) \subseteq \{\Enum, \Cau, \Sup\}$: $p$ may be used as a guard ($\Enum$),
caused ($\Cau$), or suppressed ($\Sup$). Base events are always enumerable, and
their causability and suppressability are fixed by the signature. We write
$m_\Gamma$ for the set of enumerable predicates of $\Gamma$.

\paragraph{Guardedness.}
The judgment $\Gamma \vdash \varphi : \GRD(x)^p$, for a polarity $p \in \{+,-\}$,
states that the values of $x$ satisfying $p\varphi$ can be enumerated from the
enumerable predicates (Figure~\ref{fig:app:grd}).
For a set of variables $X$, $\Gamma \vdash \varphi : \Enum_X^p$ holds iff
$\Gamma \vdash \varphi : \GRD(x)^p$ for all $x \in X$.
A trigger $(\pi, \psi)$ \emph{guards} $x$ if every disjunct $\kappa \in \pi$ contains
an atom with argument $x$, or $\Gamma \vdash \psi : \GRD(x)^+$.

\begin{figure}[h]
\footnotesize
\begin{mathpar}
\inferrule*[left=\textsc{Pred}]{p \in m_\Gamma \\ x \in \bar t}{\Gamma \vdash p(\bar t) : \GRD(x)^+}
\and
\inferrule*[left=\textsc{Eq}]{ }{\Gamma \vdash x = c : \GRD(x)^+}
\and
\inferrule*[left=\textsc{Vac}]{ }{\Gamma \vdash \top : \GRD(x)^-}
\and
\inferrule*[left=\textsc{Neg}]{\Gamma \vdash \varphi : \GRD(x)^{-p}}{\Gamma \vdash \neg\varphi : \GRD(x)^{p}}
\and
\inferrule*[left=\textsc{And}$^{+}_\mathrm{L}$]{\Gamma \vdash \varphi : \GRD(x)^+}{\Gamma \vdash \varphi\land\psi : \GRD(x)^+}
\and
\inferrule*[left=\textsc{And}$^{+}_\mathrm{R}$]{\Gamma \vdash \psi : \GRD(x)^+}{\Gamma \vdash \varphi\land\psi : \GRD(x)^+}
\and
\inferrule*[left=\textsc{And}$^-$]{\Gamma \vdash \varphi : \GRD(x)^- \\ \Gamma \vdash \psi : \GRD(x)^-}{\Gamma \vdash \varphi\land\psi : \GRD(x)^-}
\end{mathpar}
\caption{Guardedness ($c$ is a constant)}
\label{fig:app:grd}
\end{figure}

\begin{lemma}\label{lem:app:gx}
  \begin{enumerate}
  \item $m_\Gamma \vdash (\pi, \psi) \rightsquigarrow^p_x (\pi', \psi')$ holds for some $(\pi', \psi')$ iff
    every $\kappa\in\pi$ binds $x$ or $\Gamma \vdash \psi : \GRD(x)^p$. In particular, a guard for $x$
    can be extracted from a trigger iff the trigger guards $x$.
  \item $\guards^{m_\Gamma}_X(\Phi) \neq \bot$ iff $\Gamma \vdash \Phi : \Enum_X^+$.
  \end{enumerate}
  \leanref{Lemma_A_1}
\end{lemma}

\paragraph{Causation and suppression.}
The judgment $\Gamma \vdash \varphi : \alpha \triangleright \Delta$ with
$\alpha \in \{\Cau, \Sup\}$ states that the clauses $\Delta$ cause (resp.\ suppress)
$\varphi$ (Figure~\ref{fig:app:cs}). We write $\Delta^\psi$ for $\{(\pi, \varphi \land \psi) \Rightarrow \varepsilon \mid (\pi, \varphi) \Rightarrow \varepsilon \in \Delta\}$,
$\Delta[c/x]$ for the substitution of a constant $c$
for $x$ in every clause (where defined; we use a canonical constant $0$), and $\Delta{\downarrow}_x$ for
$\{(\pi', \psi') \Rightarrow \varepsilon \mid (\pi, \psi) \Rightarrow \varepsilon \in \Delta,\ m_\Gamma \vdash (\pi, \psi) \rightsquigarrow^+_x (\pi', \psi')\}$. A clause set is \emph{unconditional}
if all its clauses have the form $(\top,\top) \Rightarrow e(\bar t)$; then
$\Eventually_{[b,b]}\Delta$ and $\Next^n\Delta$ defer all effects.

\begin{figure*}[t]
\footnotesize
\begin{mathpar}
\inferrule*[left=$\top^\Cau$]{ }{\Gamma \vdash \top : \Cau \triangleright \emptyset}
\and
\inferrule*[left=\textsc{Ev}$^\Cau$]{e \in \Cau}{\Gamma \vdash e(\bar t) : \Cau \triangleright \{(\top,\top) \Rightarrow e(\bar t)\}}
\and
\inferrule*[left=\textsc{Ev}$^\Sup$]{e \in \Sup}{\Gamma \vdash e(\bar t) : \Sup \triangleright \{(e(\bar t),\top) \Rightarrow \neg e(\bar t)\}}
\and
\inferrule*[left=\textsc{Let}$^\Cau$]{\Cau \in \Gamma(e)}{\Gamma \vdash e(\bar t) : \Cau \triangleright \{(\top,\top) \Rightarrow \mathtt{Cau}\_e(\bar t)\}}
\and
\inferrule*[left=\textsc{Let}$^\Sup$]{\Sup \in \Gamma(e)}{\Gamma \vdash e(\bar t) : \Sup \triangleright \{(e(\bar t),\top) \Rightarrow \mathtt{Sup}\_e(\bar t)\}}
\and
\inferrule*[left=\textsc{Neg}$^\Cau$]{\Gamma \vdash \varphi : \Sup \triangleright \Delta}{\Gamma \vdash \neg\varphi : \Cau \triangleright \Delta}
\and
\inferrule*[left=\textsc{Neg}$^\Sup$]{\Gamma \vdash \varphi : \Cau \triangleright \Delta}{\Gamma \vdash \neg\varphi : \Sup \triangleright \Delta}
\and
\inferrule*[left=\textsc{And}$^\Cau$]{\Gamma \vdash \varphi : \Cau \triangleright \Delta_1 \\ \Gamma \vdash \psi : \Cau \triangleright \Delta_2}
  {\Gamma \vdash \varphi\land\psi : \Cau \triangleright \Delta_1 \cup \Delta_2}
\and
\inferrule*[left=\textsc{And}$^\Sup$]{\Gamma \vdash \varphi_j : \Sup \triangleright \Delta \\ \varphi_i~\text{present for}~i \neq j}
  {\Gamma \vdash \textstyle\bigwedge_i\varphi_i : \Sup \triangleright \Delta^{\bigwedge_{i \neq j}\varphi_i}}
\and
\inferrule*[left=\textsc{Ex}$^\Cau$]{\Gamma \vdash \varphi : \Cau \triangleright \Delta}{\Gamma \vdash \exists x.\,\varphi : \Cau \triangleright \Delta[0/x]}
\and
\inferrule*[left=\textsc{Ex}$^\Sup$]{\Gamma \vdash \varphi : \Sup \triangleright \Delta \\ \text{every trigger of $\Delta$ guards $x$}}
  {\Gamma \vdash \exists x.\,\varphi : \Sup \triangleright \Delta{\downarrow}_x}
\and
\inferrule*[left=\textsc{Fut}$_{\Eventually}^{\Cau}$]{\Gamma \vdash \varphi : \Cau \triangleright \Delta \\ \Delta \neq \emptyset~\text{unconditional} \\ a \le b \\ 1 \le b < \infty}
  {\Gamma \vdash \Eventually_{[a,b]}\varphi : \Cau \triangleright \Eventually_{[b,b]}\Delta}
\and
\inferrule*[left=\textsc{Fut}$_{\Next}^{\Cau}$]{\Gamma \vdash \varphi : \Cau \triangleright \Delta \\ \Delta \neq \emptyset~\text{unconditional} \\ 1 \le b}
  {\Gamma \vdash \Next_{[0,b]}\varphi : \Cau \triangleright \Next\Delta}
\and
\inferrule*[left=\textsc{Fut}$_{\Next^n}^\Cau$]{\Gamma \vdash \varphi : \Cau \triangleright \Delta \\ \Delta~\text{unconditional} \\ 1 \le n}
  {\Gamma \vdash \Next^n\varphi : \Cau \triangleright \Next^n\Delta}
\end{mathpar}
\caption{Causation and suppression}
\label{fig:app:cs}
\end{figure*}

\paragraph{Let bindings.}
Let $\LLet{e_1(\bar x_1)}{\varphi_1}{}\ldots\LLet{e_k(\bar x_k)}{\varphi_k}{}\,\Always\chi_1 \land \dots \land \Always\chi_n$
be in let-normal form, and $\chi = \chi_1 \land \dots \land \chi_n$. For each let body $\varphi_i$, let $\hat\varphi_i$ be the
formula producing its tuples ($\psi$ for $\exists\bar y.\,\psi$, where $\psi$ has no
leading existential; $\varphi_r$ for $\varphi_l \Since_I \varphi_r$; $\varphi$ for $\Prev_I\varphi$ and for
an aggregation $\bar o \gets \omega(\bar t;\bar g)\,\varphi$, whose results $\bar o$ are then excluded from $\bar x_i$).
We write $\Gamma_{<i}$ for $\Gamma$ restricted to $e_1, \ldots, e_{i-1}$, and require that
\begin{itemize}
\item $\Enum \in \Gamma(e_i)$ iff $\Gamma_{<i} \vdash \hat\varphi_i : \Enum^+_{\bar x_i\bar y}$ and, if
  $\varphi_i = \varphi_l \Since_I \varphi_r$, also $\Gamma_{<i} \vdash \varphi_l : \Enum^-_{\bar x_i}$
  (the \kw{remove} clause of the table), where $\bar x_i\bar y$ stands for $\fv(\hat\varphi_i)$ for an aggregation;
\item if $\varphi_i$ is temporal or an aggregation, then $\Enum \in \Gamma(e_i)$, and the terms $\bar t$ of an
  aggregation only use variables of $\hat\varphi_i$;
\item $\Cau \in \Gamma(e_i)$ (resp.\ $\Sup \in \Gamma(e_i)$) iff $\Enum \in \Gamma(e_i)$ and
  $\Gamma_{<i} \vdash \mathrm{cau}(\varphi_i) : \Cau \triangleright \Delta_i$ (resp.\
  $\Gamma_{<i} \vdash \mathrm{sup}(\varphi_i) : \Sup \triangleright \Delta_i$) for some $\Delta_i$,
\end{itemize}
where $\mathrm{cau}(\psi) = \mathrm{sup}(\psi) = \psi$ for present bodies,
$\mathrm{cau}(\varphi_l \Since_I \varphi_r) = \varphi_r$ if $0 \in I$ (and undefined otherwise),
$\mathrm{sup}(\varphi_l \Since_I \varphi_r) = \varphi_l \lor \varphi_r$ if $0 \in I$ and $\varphi_l$ otherwise,
and neither is defined for $\Prev_I$ or aggregations. Lets without $\Enum$ are \emph{filter-only}: they
can be used in filters, but not as guards.

\begin{definition}[EF-MFOTL]
  A formula $\Always\varphi$ is in EF-MFOTL with clause set $\Delta$ if its let-normal
  form, with lets typed by some context $\Gamma$ as above, satisfies
  $\Gamma \vdash \chi : \Cau \triangleright \Delta$.
\end{definition}

\begin{theorem}\label{thm:app:equiv}
  $\Always\varphi$ is in EF-MFOTL with clause set $\Delta$ iff the compilation of
  Section~\ref{sec:mfotl-to-ef} succeeds on $\Always\varphi$ with candidate clause set $\Delta$,
  i.e., \textsc{TypeLet} accepts all lets, and
  $\Gamma \vdash \chi \hookrightarrow^\Cau \mathcal C$ for the resulting $\Gamma$ and some $\mathcal C \ni \Delta$.\leanref{Theorem_A_2}
\end{theorem}
\begin{proof}[Proof sketch]
  The rules of Figure~\ref{fig:app:cs} mirror those of Figure~\ref{fig:rewrite}, producing a single
  clause set instead of a set of alternatives: by induction on derivations,
  $\Gamma \vdash \varphi : \alpha \triangleright \Delta$ iff
  $\Gamma \vdash \varphi \hookrightarrow^\alpha \mathcal C$ for some $\mathcal C \ni \Delta$. The side condition of
  \textsc{Ex}$^\Sup$ is equivalent to that of the rewrite rule by Lemma~\ref{lem:app:gx}(1), and the
  enumerability conditions on let bodies are equivalent to the success of $\guards$ by
  Lemma~\ref{lem:app:gx}(2). Given a typing, the guards and realizations are obtained from
  the typing derivations of the let bodies; conversely, capabilities are read off the
  guards and realizations computed by \textsc{TypeLet}.
\end{proof}

Together with Theorem~\ref{thm:compile}, every formula of EF-MFOTL whose clause set passes
the checks of Section~\ref{sec:edg} is soundly enforced by the compiled EF program.
Lemma~\ref{lem:app:gx} and Theorem~\ref{thm:app:equiv} are formalized in Lean~4, together with the
soundness of compilation and of the enforcement loop.

\begin{example}
  Consider $\varphi_\mathsf{del}$ from Example~\ref{ex:formulae}, with let-normal form
  $\Always \neg\exists d, u.~\varphi_0$ for
  $\varphi_0 = \deleteReq(d,u) \land \neg\Eventually_{[0,30]}\delete(d,u)$ and no lets. By
  \textsc{Ev}$^\Cau$ and \textsc{Fut}$_{\Eventually}^{\Cau}$,
  $\Gamma \vdash \Eventually_{[0,30]}\delete(d,u) : \Cau \triangleright
  \{(\top,\top) \Rightarrow \Eventually_{[30,30]}\delete(d,u)\}$; by \textsc{Neg}$^\Sup$ and
  \textsc{And}$^\Sup$,
  $\Gamma \vdash \varphi_0 : \Sup \triangleright
  \{(\top, \deleteReq(d,u)) \Rightarrow \Eventually_{[30,30]}\delete(d,u)\}$.
  The trigger guards $u$ and $d$ (rule \textsc{Pred}, as base events are enumerable), so two
  applications of \textsc{Ex}$^\Sup$ and \textsc{Neg}$^\Cau$ yield
  $\Gamma \vdash \neg\exists d,u.~\varphi_0 : \Cau \triangleright
  \{(\deleteReq(d,u), \top) \Rightarrow \Eventually_{[30,30]}\delete(d,u)\}$.
\end{example}
\fi

\end{document}